\documentclass[letterpaper,11pt]{scrartcl}
\usepackage[draft]{brandl}
\usepackage[margin=1in]{geometry}

\usepackage{array,xspace,hhline,colortbl,tabularx,booktabs,fixltx2e}
\usepackage{ifthen}
\usepackage{algorithm}
\usepackage{algorithmic}
\usepackage{adjustbox}
\usepackage{blkarray}% 

\newcommand{\matindex}[1]{\mbox{\scriptsize#1}}
\renewcommand{\algorithmicrequire}{\wordbox[l]{\textbf{Input}:}{\textbf{Output}:}}     		 \renewcommand{\algorithmicensure}{\wordbox[l]{\textbf{Output}:}{\textbf{Output}:}}

	\newcommand{\capa}{\mathrm{cap}}

	\newcommand{\ver}[0]{\ensuremath{\mathit{VER}}\xspace}
	\renewcommand{\ps}[0]{\ensuremath{\mathit{PS}}\xspace}
	
	\renewcommand{\dl}[0]{\ensuremath{\mathit{LD}}}
	\renewcommand{\sd}[0]{\ensuremath{\mathit{SD}}}
	\renewcommand{\pref}{\succsim\xspace}
	\newcommand{\apref}[1][]{% 
	\ifthenelse{\equal{#1}{}}{{\pref}_N}{{\pref^{#1}}_N}% 
	}
	\newcommand{\aprefs}[1][]{% 
	\ifthenelse{\equal{#1}{}}{{\succ}_N}{{\succ^{#1}}_N}% 
	}
	\newcommand{\opref}[1][]{% 
	\ifthenelse{\equal{#1}{}}{{\pref}_O}{{\pref^{#1}}_O}% 
	}
	\newcommand{\oprefs}[1][]{% 
	\ifthenelse{\equal{#1}{}}{{\succ}_O}{{\succ^{#1}}_O}% 
	}
	\newcommand{\p}{(\pref_N,\pref_O)}
	\newcommand{\pss}[1][]{% 
	\ifthenelse{\equal{#1}{}}{({\succ}_N,{\succ}_O)}{({\succ^{#1}}_N,{\succ^{#1}}_O)}% 
	}

    \makeatletter
\def\@fnsymbol#1{\ensuremath{\ifcase#1\or \dagger\or \ddagger\or
   \mathsection\or \mathparagraph\or \|\or **\or \dagger\dagger
   \or \ddagger\ddagger \else\@ctrerr\fi}}
    \makeatother

\begin{document}

\title{% 
The Vigilant Eating Rule:\\ A General Approach for Probabilistic Economic Design with Constraints
}

\author{Haris Aziz\thanks{UNSW Sydney and Data61 CSIRO, Australia, \texttt{haris.aziz@unsw.edu.au}} \and Florian Brandl\thanks{University of Bonn, Germany, \texttt{florian.brandl@uni-bonn.de}}}

	\maketitle

\begin{abstract}
	We consider the problem of probabilistic allocation of objects under ordinal preferences.  We devise an allocation mechanism, called the vigilant eating rule (VER), that applies to nearly arbitrary feasibility constraints. It is constrained ordinally efficient, can be computed efficiently for a large class of constraints, and treats agents equally if they have the same preferences and are subject to the same constraints. When the set of feasible allocations is convex, we also present a characterization of our rule based on ordinal egalitarianism. Our results about VER do not just apply to allocation problems but to all collective choice problems in which agents have ordinal preferences over discrete outcomes. As a case study, we assume objects have priorities for agents and apply VER to sets of probabilistic allocations that are constrained by stability. VER coincides with the (extended) probabilistic serial rule when priorities are flat and the agent proposing deterministic deferred acceptance algorithm when preferences and priorities are strict. While VER always returns a stable and constrained efficient allocation, it fails to be strategyproof, unconstrained efficient, and envy-free. We show, however, that each of these three properties is incompatible with stability and constrained efficiency.
\end{abstract}

\section{Introduction}

The theory and application of allocation and matching mechanisms have proved to be one of the major success stories of algorithmic economics. An ongoing challenge is designing mechanisms that can handle complex constraints arising in new applications. In this paper, we present a versatile and robust allocation rule that achieves fair and efficient outcomes for a plethora of economic design problems including matching market design. 
It simultaneously generalizes celebrated rules in the literature including the probabilistic serial rule and the deferred acceptance algorithm.

We study the problem of allocating a set of indivisible objects to a group of agents based on the agents' preferences over objects.
The objects could be seats at schools, dormitory rooms, job placements, or kidney transplants for example.
In these applications, it is important, perhaps even mandatory, that allocations are fair in the sense that no agent (justifiably) envies some other agent. 
Since objects are indivisible, envy-free allocations may however not exist.
One possible remedy for this dilemma is to consider randomizations over deterministic allocations, which can restore fairness ex-ante.
This approach is deeply rooted in the literature on resource allocation and has gained popularity in recent years.
The seminal works of \citet{HyZe79a} and \citet{BoMo01a} provide allocation mechanisms for the case when each agent demands exactly one object and preferences are given by linear utility functions or ordinal preferences over objects, respectively.
We shall be concerned with ordinal preferences under two generalizations: first, agents may receive more than one object, and, second, there may be constraints on which random allocations are feasible.

\citeauthor{BoMo01a}'s \emph{probabilistic serial rule} can be pictured as follows.
Time runs continuously from 0 to 1. 
At time 0, each agent starts off by eating her most-preferred object.
An object becomes unavailable once the cumulative time agents spent eating it equals 1. 
Whenever an object becomes unavailable, the agents who have been eating it switch to their next most-preferred object among those which are still available.
The probability with which an agent receives an object in the final random allocation equals the time she spent eating that object.
The Birkhoff-von Neumann Theorem ensures that these probabilities can be attained by randomizing over deterministic allocations.
The probabilistic serial rule enjoys several appealing properties: it is weakly strategyproof, can be computed efficiently, and always yields an allocation that is efficient and envy-free (when the agents' preferences over probabilistic allocations are based on stochastic dominance).

A tacit assumption in the formulation of the probabilistic serial rule is that all allocations are feasible.
But several applications of object allocation require allocation rules to respect various feasibility constraints.
For example, allocating students to courses may be subject to curricular constraints that require students to take a minimal number of courses or courses in different subjects.
Likewise, when allocating donor's kidneys to patients, logistics and blood type compatibility impose constraints on which exchanges are feasible.
In these examples, one has constraints on deterministic allocations and stipulates that a probabilistic allocation is feasible if it can be written as a convex combination of deterministic allocations that meet the constraints. 
We will allow for a more general class of constraints where the primitive is a set of feasible \emph{probabilistic} allocations.
This of course includes the previous example.\footnote{In this case, the set of feasible probabilistic allocations is the face of the simplex of probabilistic allocations spanned by deterministic allocations satisfying the constraints.}
More generally, it can capture ex-ante constraints on the probabilistic allocation and possibly different constraints on the ex-post allocation.
For example, a system designer may impose ex-ante stability constraints in addition to ex-post constraints.
Then a probabilistic allocation is feasible if it satisfies the ex-ante constraints and can be decomposed into deterministic allocations that meet the ex-post constraints \citep[see, for example,][]{ASS19a,AkNi20a}.

Our main contribution is a generalization of the probabilistic serial rule, called the \emph{vigilant eating rule (VER)}, which can handle multi-object allocation under nearly arbitrary constraints. 
Formally, it requires that the set of feasible probabilistic allocations is closed. 
Previous generalizations of the probabilistic serial rule required that feasible probabilistic allocations are given by constraints that form a bi-hierarchical structure \citep{BCKM12a}, constitute lower and upper quotas \citep{ASS19a,ASS20a}, or exclude a fixed set of deterministic allocations \citep{AzSt14a}.
In each of these three domains, VER coincides with the proposed generalizations of the probabilistic serial rule.
Despite being generally applicable, VER retains many of the properties of the probabilistic serial rule: it always yields an allocation that is efficient among feasible allocations, gives the ex-ante same allotment to agents who have the same preferences and are subject to the same constraints, and can be computed efficiently whenever the set of feasible allocations is a union of polytopes (described by polynomially many linear constraints).
Clearly, no rule can be unconstrained efficient or envy-free without restrictions on the set of feasible allocations.

\medskip

A class of problems that falls under the umbrella of object allocation is two-sided allocation, where, in addition to preferences, agents have different priorities for objects.
The literature on two-sided allocation is well-developed \citep[see, for example,][]{RoSo90a}. 
Almost universally, allocations are required to be stable: if an agent prefers an object to the one she has been allocated, the former object is given to an agent who has a higher priority for it.
We will view stable two-sided allocation as a special case of allocation under constraints, where only stable allocations are feasible. 
Unlike when constraints come from quotas, stability constraints depend on the preferences of the agents and are thus harder to handle. 

In the classical formulation of two-sided allocation as marriage markets, preferences and priorities are strict and one considers only deterministic mechanisms. 
The most prominent representative is the (agent proposing) deferred acceptance rule of \citet{GaSh62a}.
It enjoys several appealing properties such as strategyproofness and Pareto efficiency (when interpreting priorities as preferences).
Yet, applications like school choice where priorities are typically coarse, that is, contain large indifference classes, motivate the search for rules that can deal with this more general problem domain.
There are various reasons for considering probabilistic rules when priorities are weak. 
Perhaps the most evident one is that it allows for fairness ex-ante, which is not always attainable when deterministically allocating objects.
Take the example of two agents preferring object $a$ over object $b$. 
If both agents have the same priority for $a$ and for $b$, allocating $a$ to either agent with a probability of $50\%$ is fair ex-ante, while no deterministic allocation is fair.

We propose VER as a promising probabilistic mechanism for stable two-sided allocation when preferences and priorities are weak. 
Stability can be generalized to probabilistic allocations in various ways. 
We consider four different notions of stability that have been considered in the literature and all reduce to pairwise stability for deterministic allocations.
For each stability notion, we can apply VER to the set of stable allocations and thereby obtain a mechanism that always yields stable and constrained efficient allocations.
Each of these mechanisms provides a natural transition between the deferred acceptance algorithm and the probabilistic serial rule.
It coincides with the former if priorities are strict and with the latter if priorities are completely flat (and preferences are strict in both cases).

\paragraph{Contributions}

We explore the space of desirable rules and algorithms for probabilistic allocation. 
In contrast to previous work on probabilistic rules for stable allocation, we allow the agents to have weak preferences.  
Our main contribution is the \emph{Vigilant Eating Rule (VER)}, which applies to any class of distributional constraints, including ones given by stability notions. 
For example, it applies to \emph{non}-bi-hierarchical constraints, which are not handled by the generalization of the probabilistic serial rule due to \citet{BCKM12a}. 
It also applies to non-convex constraints arising from integrality requirements or some stability notions. 
Several stability concepts are not necessarily convex. Examples include ex-ante stability~\citep{KeUn15a}, respect of priorities~\citep{Dela20a}, and a fractional stability concept considered by \citet{CFKV19a}.
We show that the outcome of VER is efficient within the set of feasible allocations with respect to stochastic dominance (that is, \sd-efficient). 
If the feasible set is convex, the outcome is ordinally egalitarian in the sense defined by \citet{Bogo15b} \citep[see also][]{Bogo18a}. 
Our results about VER are not restricted to allocation problems.
They apply to collective decision making in general whenever agents have ordinal preferences over deterministic outcomes. 

For object allocation under priorities, we study VER in more detail and argue that it leads to compelling allocation rules. 
VER returns a stable allocation that is constrained \sd-efficient within the set of stable allocations. 
To the best of our knowledge, our algorithms are the first rules for constrained \sd-efficient probabilistic allocation under priorities when the agents' preferences are weak.\footnote{Some problems invariably involve ties in the agents' preferences. For example, when allocating scarce medical resources such as ventilators where different categories of resources have priorities over agents, agents are indifferent between all ventilator units~\citep{PSU+20a}.}
For several properties that VER does not satisfy, we prove that they are incompatible with stability, constrained \sd-efficiency, or their conjunction. 
In particular, we prove that even for the weakest stability notion we consider (claimwise stability), the set of stable probabilistic allocations can be disjoint from the set of \sd-efficient allocations and the set of weakly \sd-envy-free allocations.  
For each of our stability notions, no rule can be weakly \sd-strategyproof and always yield a stable and constrained (within the set of stable allocations) \sd-efficient allocation. 

More broadly, we present a generalization of the probabilistic serial rule that can handle a more general class of constraints than previous generalizations. 
In restricted domains, VER coincides with well-established rules (see Table~\ref{table:summaryVER}). 
We also show that VER can also be useful when no randomization is allowed.

As we endeavor to apply the theory of fair allocation and market design to facilitate applications, it is critical to develop robust mechanisms that can flexibly handle domain-specific constraints. 
Our approach is especially well-suited for settings with ordinal preferences. We hope that the framework, methodology, and rule we present can be useful for new application domains.

\begin{table*}[tb]
\label{tab:compare} 
\centering
\makebox[\textwidth][c]{
\begin{adjustbox}{max width=1.1\textwidth}
\begin{tabular}{llll}
\toprule
Constraints& Preferences & Priorities &Resulting Rule \\
\midrule
--& strict & -- & probabilistic serial (PS) \citep{BoMo01a}\\
--& weak & -- & extended PS \citep{KaSe06a}\\
\sd-individual rationality\footnotemark &strict & -- & controlled consuming \citep{AtSe11a}\\
upper quotas\footnotemark &strict & -- & generalized PS \citep{BCKM12a}\\
lower \& upper quotas & strict & -- & generalized PS \citep{ASS20a} \\
feasible ex-post allocations & weak & -- & egalitarian simultaneous reservation \citep{AzSt14a}\\
feasible ex-post allocations & strict & -- &generalized constrained PS \citep{Balb19b} \\
ex-post stability & strict  & strict & agent proposing DA \citep{GaSh62a}\\
claimwise stability & strict & weak & constrained PS \citep{Afac18a}\\
--&dichotomous& dichotomous & egalitarian rule~\citep{BoMo04a}\\
integral allocations; unit demand&strict& -- & serial dictatorship\\
--&strict over bundles& -- &PS for bundles~\citep{ChLi20a}\\
  \bottomrule
\end{tabular}
\end{adjustbox}
}
\caption{Overview of generalizations of the probabilistic serial rule to particular constraints and preference and priority structures.
VER coincides with each of these rules in the corresponding domain.}
\label{table:summaryVER}
\end{table*}

\addtocounter{footnote}{-2}
\stepcounter{footnote}\footnotetext{The \sd-individual rationality constraints require that each agent's allocation stochastically dominates her initial endowment. The endowments are part of the specification of an allocation instance in addition to the agents' preferences.}
\stepcounter{footnote}\footnotetext{\citet{BCKM12a} assume that the constraints can be written as the union of two hierarchies of constraints. A collection of constraints forms a hierarchy if for any two constraints, the set of agent-object pairs they apply to are either disjoint or one is contained in the other.}

\section{Related Work}

We explore two-sided allocation while relaxing the assumption that allocations have to be integral or deterministic. 
In our paper, an allocation specifies, for each pair of agent and object, the probability that the agent receives that object.
This probabilistic approach is important on several accounts.
We refer the reader to \citet{Aziz18a} for an in-depth discussion.
In short, randomization allows one to achieve minimal fairness requirements such as equal treatment of equals that are unobtainable through deterministic allocations.
Secondly, probabilistic allocations can also be interpreted as time-sharing or fractional allocations.

\paragraph{Probabilistic matching market design}
Although most of the work on two-sided allocation has focussed on deterministic allocations, some authors have pointed out connections between integral stable allocations and their linear relaxations (see, for example, \citet{RRV93a, TeSe98a}). 
\citet{ErEr08a,ErEr17a} and \citet{Kest10a} undertook a study of school choice when schools have weak priorities. 
They highlighted that tie-breaking can lead to a loss of efficiency. 
However, they focus on achieving constrained efficiency of deterministic or integral allocations. 
Similarly, \citet{ANR19a} consider the impact of various types of tie-breaking on the efficiency of the allocations. \citet{KaKo15a,KaKo17a,KaKo17b} and \citet{KTY18a} consider deterministic two-sided matching under general distributional constraints.  
These constraints capture several real-life scenarios, such as restrictions on the number of doctors in a particular region.
As noted before, deterministic outcomes cannot guarantee ex-ante equal treatment of equals or capture time-sharing. 
Deterministic allocation also has several other aspects that are different from its probabilistic counterpart. 
For example, the problem of computing deterministic allocations that satisfy a set of distributional constraints is NP-hard for many classes of constraints \citep[see, for example,][]{BFIM10a}. 
Considering a probabilistic allocation and then suitably rounding it can be an indirect, but computationally more tractable approach to arrive at a deterministic allocation (see, for example, \citet{AkNi20a}).

\citet{KeUn15a} initiated a serious study of stability notions and mechanisms for probabilistic allocation under weak priorities. 
They focussed on a strong version of ex-ante stability and proposed two mechanisms. 
First, they present the \emph{fractional deferred acceptance} algorithm, which returns a strongly ex-ante stable allocation. 
They then modify it to derive the \emph{fractional deferred acceptance and trading} algorithm that returns a strongly ex-ante stable allocation which is also Pareto efficient with respect to stochastic dominance among stable allocations. 
\citet{Han17a} presented a similar algorithm that returns an ex-ante stable allocation.

\citet{Afac18a} considered a more general model in which objects have probabilities for prioritizing one agent over another. 
He proposed a weak stability notion called \emph{claimwise stability} and showed that no weakly strategyproof mechanism always returns an allocation that is claimwise stable and efficient among claimwise stable allocations.
Both strategyproofness and efficiency assume that preferences are based on stochastic dominance. 
\citeauthor{Afac18a} presented his \emph{constrained probabilistic serial} algorithm, which returns an allocation that is constrained efficient among  claimwise stable allocations and can be computed in polynomial time. 

All of the algorithms that return constrained efficient and stable probabilistic allocations discussed above assume that agents have strict preferences \citep{Afac18a,Han17a,KeUn15a}.
We will not make this assumption and our main results are valid for weak preferences. Furthermore, our approach is general enough to allow for any type of feasibility constraints. 

When considering probabilistic stable allocations, there are various notions of stability that coincide with the classical pairwise stability notion for deterministic allocations. \citet{AzKl19a} investigate the taxonomy of stability notions for probabilistic allocations and map out the logical relationships between them.

\citet{CFKV19a} considered fractional stable allocations under cardinal utilities. 
A fractional allocation is stable if no pair gets higher utility in being integrally matched to each other than under the fractional allocation. 
They explore computational problems of finding stable fractional allocations with high welfare. 
\citet{HMPY18a} also considered cardinal utilities and presented a pseudo-market mechanism that simultaneously generalizes the deferred acceptance algorithm and the mechanism of \citet{HyZe79a}. 
It returns an allocation that is ex-ante stable. 
\citet{EMZ19b} have extended the pseudo-market rule to handle constraints on deterministic allocations.

\paragraph{Extensions of probabilistic serial without priorities}

There are several papers on probabilistic allocation when the objects have no priorities. 
Our model allows for weak priorities and, thus, has no or flat priorities as a special case.
One of the seminal works on random allocation under ordinal preferences without priorities is by \citet{BoMo01a}, who compared the random priority rule with the probabilistic serial rule. 
They showed that the probabilistic serial rule is \sd-efficient and \sd-envy-free, whereas the random priority rule satisfies neither property. 

The attractive properties of the probabilistic serial rule have led to a whole line of work on extensions to more general settings. 
Those include weak preferences~\citep{KaSe06a}, multi-unit demands~\citep{Koji09a}, endowments~\citep{AtSe11a,Yilm10a,YuZh20a}, bundled allocations~\citep{ChLi20a} and the probabilistic voting setting~\citep{AzSt14a}. 
VER coincides with each of these rules on their domain. 
It captures and generalizes the key insights that underlie the probabilistic serial rule. 
Moreover, our rule provides a conceptually easier formulation of the previously presented algorithms for specific domains~\citep{KaSe06a,Yilm10a,AtSe11a}.

\citet{BCKM12a} consider probabilistic allocation without priorities.
They presented a generalization of the probabilistic serial rule which allows for distributional constraints that are given by upper quotas and form a bi-hierarchy. 
\citet{FSZ18a} extend the probabilistic serial rule to submodular constraints.
In their algorithms, agents continue eating a most preferred object until some upper quota on a set of agent-object pairs is met. 
This approach works if there are upper quotas on agent-object groups that form a bi-hierarchical constraint structure; 
it does however not extend to lower quotas on agent-object groups or constraint structures that do not form a bi-hierarchy.   
 \citet{ASS19a,ASS20a} consider probabilistic allocation with lower and upper bounds on the types of students for each school without priorities and give a suitable generalization of the probabilistic serial rule. 
 The class of constraints to which VER applies includes those considered by \citet{BCKM12a}, \citet{FSZ18a}, and \citet{ASS19a,ASS20a}. 
In particular, VER can efficiently handle any kind of linear constraints. 
Our methodology also has some important differences with that of \citet{ASS19a,ASS20a}. 
To deal with non-convex feasible sets, our rule is based on a priority order over agents. 
An agent can only eat more of an object if this does not prevent the agents preceding her in the order from eating more. 
While agents earlier in the order have priority, no requirement is put on how much they can eat except that it is a non-zero amount. 
These considerations are necessary to apply the eating approach to general classes of feasibility constraints.
When the set of feasible allocations is convex, the outcome is independent of the priority order over agents.

Our model allows for constraints both on the set of feasible ex-post allocations as well as constraints on the ex-ante allocation. Note that if we are only concerned about constraints on the set of ex-post allocations, we can explicitly list the feasible ex-post allocations as the set of alternatives, and run the egalitarian simultaneous reservation rule of \citet{AzSt14a}, which is constrained \sd-efficient and has been referred to as the appropriate extension of the probabilistic serial rule to the social choice domain~\citep{Bogo18a}.\footnote{More recently, \citet{Balb19b} proposed an adaptation of the probabilistic serial rule for strict preferences, called the \emph{generalized constrained probabilistic serial rule}, where the set of feasible ex-post allocations are enumerated. 
Algorithmically, the results of \citet{Balb19b} can be achieved by enumerating the feasible ex-post allocations and running the egalitarian simultaneous reservation algorithm of \citet{AzSt14a}, which seamlessly handles weak preferences.} However, in allocation problems, there are several advantages of VER. 
Firstly, the constraints themselves may be ex-ante. 
For example, three of the four stability notions we consider are defined by constraints on the ex-ante allocation matrix and not by enumerating the set of feasible ex-post allocations. 
Another example is \emph{time-sharing} where the focus is on the fractional allocations and not on lotteries over deterministic allocations.
Secondly, even if certain ex-ante constraints can be captured by enumerating ex-post allocations, the approach may be computationally prohibitive as one may need to enumerate an exponential number of allocations.

\section{The Model}

Let $N$ be a set of $n$ agents and $O$ be a set of $m$ objects.
Each agent $i\in N$ has a weak preference relation (complete, reflexive, and transitive) $\pref_i$ over objects.
The symmetric and asymmetric parts of a relation $\pref$ are denoted by $\sim$ and $\succ$ as usual.
A preference relation is strict if it is a linear order; otherwise it is weak.
By $\mathcal E_i = O/{\sim_i}$ we denote the collection of equivalence classes of objects induced by $\pref_i$.
Abusing notation, we also write $\pref_i$ for the relation on $\mathcal E_i$ induced by $\pref_i$.
A preference profile $\apref = ({\pref_i})_{i\in N}$ consists of preferences for each agent.

A probabilistic allocation of objects to agents is a matrix $p\in \mathbb R_+^{N\times O}$ such that all columns sum to at most 1.
We will write $p(i)$ for the allotment of agent $i$, that is, the row of $p$ corresponding to agent $i$.
If $O'\subset O$, $p(i,O') = \sum_{o\in O'} p(i,o)$ is the probability with which $i$ receives an object in $O'$.
We say that $p$ is complete if all columns sum to 1 meaning that all objects are fully allocated; otherwise it is partial.
If $p\in \{0,1\}^{N\times O}$, it is a deterministic allocation.
Every probabilistic allocation can be written as a lottery over deterministic allocations.\footnote{In fact, one can use the Birkhoff-von Neumann Theorem to show the following: if $p$ is a probabilistic allocation with $p(O) = \alpha_i$ for all $i$, then $p$ can be written as a lottery over deterministic allocations in which agent $i$ receives either $\lfloor\alpha_i\rfloor$ or $\lceil\alpha_i\rceil$ objects.}
Whenever left unspecified, allocations are assumed to be complete and not necessarily deterministic.
An allocation $p$ satisfies equal treatment of equals for a profile $\apref$ if ${\pref_i} = {\pref_j}$ implies $p(i,E) = p(j,E)$ for all $E\in\mathcal E_i = \mathcal E_j$.

We introduce two extensions of preferences over objects to preferences over probabilistic allocations. 
Let $\pref$ be a preference relation over $O$ and $x,y\in\mathbb R^O_+$.
We say that $x$ (first-order) stochastically dominates $y$ with respect to $\pref$, written $x\pref^{\sd} y$, if
\begin{align*}
	\sum_{o'\pref o} {(}x(o') - y(o'){)}\ge 0 \text{ for all } o \in O\tag{stochastic dominance}
\end{align*}
If at least one inequality is strict, we write $x\succ^{\sd} y$.
Secondly, $x$ lexicographically dominates $y$ with respect to $\pref$, written $x\succ^{\dl} y$, if
\begin{align*}
	\text{there is $E\in \mathcal E$ with } x(E) > y(E)\text{ and } x(E') = y(E')\text{ for all } E'\in\mathcal E \text{ with } E'\succ E
	\tag{lexicographic dominance}
\end{align*}
If $x(E) = y(E)$ for all $E\in\mathcal E$, we have $x \sim^\dl y$ and $x \sim^{\sd} y$ and simply write $x \sim y$.
Notice that lexicographic dominance refines stochastic dominance.

We assume that each agent's preferences over allocations only depend on her own allotment.
That is, for allocations $p$ and $q$, $p \pref_i^\sd q$ if and only if $p(i) \pref_i^\sd q(i)$.
Moreover, $p$ stochastically dominates $q$ if $p(i)\pref_i^{\sd} q(i)$ for all $i\in N$ and at least one agent has a strict preference; $p$ is \sd-efficient if it is not stochastically dominated by any allocation.
When $X\subset\mathbb R_+^{N\times O}$ is a set of allocations, we say that $p$ is constrained \sd-efficient for $X$ if it is not stochastically dominated by any allocation in $X$.
The analogous definitions apply when preferences are extended based on lexicographic dominance.
Since preferences based on lexicographic dominance refine preferences based on stochastic dominance, $\dl$-efficiency is more restrictive than $\sd$-efficiency.
If $p(i) \sim_i q(i)$ for all $i$, we say that all agents are indifferent between $p$ and $q$.

An allocation mechanism $f$ maps a preference profile $\apref$ and a set of feasible (complete) allocations $X$ to an allocation $f(\apref,X)\in X$.
For any property for allocations, we say that $f$ has this property if $f(\apref,X)$ has the property for all preference profiles $\apref$ and sets $X$.
For example, $f$ is \sd-efficient if $f(\apref,X)$ is \sd-efficient for all $\apref$ and $X$.

\section{Vigilant Eating Rules}\label{sec:ver}

The probabilistic serial rule, sometimes also called the eating algorithm, was introduced by \citet{BoMo01a} and roughly works as follows.
Time runs continuously from 0 onward. 
At time 0, each agent starts off by eating her most-preferred object.
An object becomes unavailable if the cumulative time agents spent eating it equals 1. 
Whenever an object becomes unavailable, the agents who have been eating it switch to their next most-preferred object among those which are still available.
The algorithm terminates when all objects are unavailable (in which case the probabilities for each object sum to 1).
The probability with which an agent receives an object in the final allocation equals the time she spent eating it.

The description of the eating rule presumes that the agents always have a unique favorite object and thus strict preferences.
If an agent is indifferent between multiple objects, one can think of eating from that equivalence class as increasing the probability of getting some object from that class.
Another assumption is that all allocations are feasible. 
The challenge in defining eating rules when not all allocations are feasible is to not let agents eat objects if this would make it impossible to continue eating in a way that results in a feasible allocation.
In other words, we need to ensure that the partial allocation at each time can be extended to a feasible allocation.
Even if the set of feasible allocations consists of all complete allocations satisfying upper bounds on the probabilities, this cannot always be achieved by simply letting agents eat objects as long as this does not violate any of the bounds.
To see this, consider the following

\begin{example}\label{example:nonbi}
	Assume two agents have the preferences depicted below.
	\begin{center}
	\begin{tabular}{lccccc}
	$\succ_1$:&$a$&$b$\\
	$\succ_2$:&$a$&$b$\\
	\end{tabular}
	\end{center}
	Suppose we want to get an allocation $p$ in which each agent gets one unit of objects: $p(1,a)+p(1,b)=1$ and $p(2,a)+p(2,b)=1$.
Suppose additionally that there is a constraint that $p(1,a)+p(2,b)\leq 1/2$.
If agents start eating their most preferred object $a$, they each get $1/2$ of $a$. 
Up to this point, no constraint is violated. 
However, there is no completion of the partial allocation that satisfies the constraints. 
The example shows that the Generalized Probabilistic Serial Rule of \citet{BCKM12a} cannot handle the above constraints.\footnote{\citet{BCKM12a} defined the Generalized Probabilistic Serial Rule for bi-hierarchical constraint structures. The constraints in \Cref{example:nonbi} do not form a bi-hierarchy.}
\end{example}

\subsection{Vigilant Eating}

We thus need a more ``vigilant'' approach.
The idea is to let an agent eat an object only if the resulting partial allocation can be extended to a complete allocation.
In general, it could be however, that agents $i$ and $j$ both can individually eat object $o$, but not simultaneously.
Thus, which objects an agent can eat may depend on which objects other agents eat (and these dependencies may be cyclic).
We cope with this issue by introducing a priority order over agents. 
\Cref{algo:closedver} formalizes our approach.
It is formulated for weak preferences, so instead of eating objects, agents increase their probability for objects in the equivalence classes induced by their preferences.
We keep track of how much of every equivalence class each agent is guaranteed, while always ensuring that there is a feasible allocation of objects which meets all guarantees.

Fix some priority order over agents, say, the lexicographic order.
Start with setting all guarantees to 0.
At the beginning of each round, we decide for every agent for which equivalence class she gets to increase her guarantee.
An equivalence class $E$ is available to agent 1 if there exists a feasible allocation that meets all of the previous guarantees and assigns a strictly higher probability for $E$ to $i$ than her current guarantee for $E$.
Agent 1 gets to increase her probability for her most preferred equivalence class $E_1$ among those available to her.
In general, $E$ is available to agent $i$ if there exists a feasible allocation that meets all guarantees established in the previous rounds, assigns a higher probability for $E_j$ to $j$ for each $j < i$, and assigns a higher probability for $E$ to $i$ than the current guarantee. 
Agent $i$ gets to eat her favorite equivalence class $E_i$ among those available to her.
Agents can increase their probabilities for the classes $E_i$ for as long as all the resulting guarantees can be met by some feasible allocation. 
More precisely, if $\pi^k$ stores the guarantees at the beginning of round $k$, find the largest $\delta$ so that there exists a feasible allocation which meets all guarantees imposed by $\pi^k$ and assigns a probability of at least $\pi^k(i,E_i) + \delta$ for $E_i$ to $i$ for all $i$.
If $\delta^*$ is this maximal value, we let $\pi^{k+1}$ be array that increases $i$'s guarantee for $E_i$ by $\delta^*$ compared to $\pi^k$ and is otherwise identical to $\pi^k$.
After at most $m \cdot n$ rounds, no guarantee can be further increased and the process terminates.
By construction, $X$ contains an allocation that meets all of the guarantees and all agents are indifferent between all such allocations.
Hence, it is irrelevant which of those we choose as the final allocation.

We call this mechanism the \emph{vigilant eating rule (VER)}.
Whenever the set of feasible allocations $X$ is non-empty and closed, VER results in a feasible allocation that is constrained \sd-efficient with respect to $X$ (see Theorems~\ref{thm:verterminates} and~\ref{thm:constrainedefficient}).
For feasible sets that are not closed, constrained \sd-efficient allocations may not exist, which shows that our requirement is minimal.
If $X$ is a union of polytopes, it can be computed in polynomial time in number of inequalities used to describe $X$ (\Cref{thm:polynomial}). 
Moreover, if $X$ is convex, the outcome of VER is ordinally egalitarian (\Cref{thm:ordinalegal}) and independent of the priority order over agents.
If $X$ meets certain symmetry conditions, VER satisfies equal treatment of equals (\Cref{thm:symmetric}).
Lastly, we explore an extension of VER that allows agents to eat at different rates (\Cref{thm:dl}).

\begin{algorithm}
	\caption{The Vigilant Eating Rule}
	\label{algo:closedver}

	\renewcommand{\algorithmicrequire}{\wordbox[l]{\textbf{Input}:}{\textbf{Output}:}}
	\renewcommand{\algorithmicensure}{\wordbox[l]{\textbf{Output}:}{\textbf{Output}:}}
	\hspace*{\algorithmicindent} \textbf{Input:} A preference profile $\apref$ and a non-empty and closed set $X$ of allocations\\
	\hspace*{\algorithmicindent} \textbf{Output:} An allocation $p\in X$
	\algsetup{linenodelimiter=\,}
	  \begin{algorithmic}[1]
	\STATE $k\longleftarrow 0$ (the round of the algorithm)
	\STATE $N'\longleftarrow N$ (the set of active agents)
	\STATE $\pi^0(i,E)\longleftarrow 0$ for all $i\in N$ and $E\in \mathcal E_i$ 
	\WHILE{$N'\neq \emptyset$}\label{ln:while}
	\FOR{$i \in N'$ in increasing order}
		\STATE $\mathcal F_i \longleftarrow \{E\in \mathcal E_i\colon \text{there is } p\in X  \text{ so that } p(j,E') \ge \pi^k(j,E') \text{ for all } j\in N \text{ and } E'\in \mathcal E_j\text{, } p(i,E) > \pi^k(i,E)\text{, and } p(j,E_j) > \pi^k(j,E_j)\text{ for all $j\in N' $ with $j<i$}\}$\label{ln:fi}
		\IF{$\mathcal F_i \neq\emptyset$}
		\STATE $E_i \longleftarrow \max_{\pref_i} \mathcal F_i$\label{ln:oi}
		\ELSE
		\STATE $N' \longleftarrow N'\setminus \{i\}$
		\ENDIF
	\ENDFOR
	\STATE Compute $\delta^*$ as follows
	\begin{align*}
	\delta^*&=\max_{\delta,p} \delta  \text{ s.t. }& \\
	p(i,E_i)&\ge  \pi^k(i,E_i) + \delta&\forall i\in N' &&\text{(agent $i$ consumes $E_i$)}\\
	p(i,E)&\ge \pi^k(i,E) &\forall i\in N\text{ and } E\in \mathcal E_i &&\text{($p$ extends $\pi^k$)}\\
	p&\in X &&&\text{($p$ is feasible)}
	\end{align*}\label{ln:delta}
	\vspace{-1em}
	\STATE $\pi^{k+1}(i,E)\longleftarrow  \pi^{k}(i,E)$ for all $i\in N$ and $E \in \mathcal E_i$\label{ln:pk+1b}
	\STATE $\pi^{k+1}(i,E_i)\longleftarrow \pi^{k}(i,E_i)+\delta^*$ for all $i\in N'$\label{ln:pk+1a}
	\STATE $k\longleftarrow k+1$
	\ENDWHILE
	\RETURN $p$ (as computed in Line~\ref{ln:delta}) 
	\end{algorithmic}
\end{algorithm}

\begin{example}[Illustration of VER]
	Let us revisit Example~\ref{example:nonbi}, where we had two agents with the following preferences.
	\begin{center}
	\begin{tabular}{lccccc}
	$\succ_1$:&$a$&$b$\\
	$\succ_2$:&$a$&$b$\\
	\end{tabular}
	\end{center}
	As before, we want to get an allocation $p$ satisfying the row constraints $p(1,a)+p(1,b)=1$ and $p(2,a)+p(2,b)=1$ and the diagonal constraint $p(1,a)+p(2,b)\leq 1/2$.
We illustrate how VER works in this example.
The order of consumption of objects over time is illustrated in Figure~\ref{figure:ver}. 

\begin{itemize}
	\item Round 1: In the first round, we identify the most preferred object for agent $1$ that she can consume. 
	That object is $a$. 
Next, we find the most preferred object for agent $2$ that she can consume given that agent 1 receives $a$ with positive probability. 
That object is also $a$. 
The algorithm computes the maximum amount that the agents can consume from $a$ while ensuring that there exists a feasible allocation extending the current partial allocation. 
This amount is $1/4$ since if agent $1$ were to consume more of object $a$, the row constraints would force agent 2 to consume more than $1/4$ of object $b$.
A violation of the diagonal constraint would thus be inevitable.
Hence, agent 1 has move on to consuming object $b$. 
At the end of the round, we have the following guarantees: $\pi^1(1,a)\geq 1/4$ and $\pi^1(2,a)\geq 1/4$.
\item Round 2: 
In the second round, agent $1$ consumes $1/2$ of object $b$ whereas agent $2$ consumes $1/2$ more of object $a$. 
At this point, object $a$ becomes unavailable and agent $2$ has to move on to consuming object $b$. 
At the end of the round, we have the following guarantees: $\pi^2(1,a)\geq 1/4$, $\pi^2(1,b)\geq 1/2$, and $\pi^2(2,b)\geq 3/4$. 
\item Round 3: Now both agents consume $1/4$ more of $b$ until $b$ becomes unavailable. 
At the end of round 3, we have the following guarantees: $\pi^3(1,a)\geq 1/4$, $\pi^3(1,b)\geq 3/4$, $\pi^3(2,a)\geq 3/4$ and $\pi^3(2,b)\geq 1/4$.
Since both objects are unavailable, the algorithm terminates with the following allocation satisfying all guarantees and the original constraints. 
\begin{center}
		$p=$
					 \begin{blockarray}{ccccccccc}
					 		&&\matindex{$a$}&\matindex{$b$}&\\
					 	    \begin{block}{c(cccccccc)}
					 			\matindex{$1$}& &$\nicefrac14$&$$\nicefrac34$$& \\
					 		\matindex{$2$}&&$\nicefrac34$&$$\nicefrac14$$& \\
					 	    \end{block}
					 	  \end{blockarray}

\end{center}

\begin{figure}
      	\begin{center}
      	             \begin{tikzpicture}[scale=0.3]
      	                 \centering
      	                 \draw[-] (0,0) -- (0,4);
      	                 \draw[-] (0,0) -- (20,0);

      	                 \draw[-] (20,4) -- (20,0);

      	\draw[-] (0,2) -- (20,2);
      	\draw[-] (0,4) -- (20,4);

      	\draw[-] (0,4) -- (20,4);

 	\draw[-] (5,0) -- (5,4);

      	\draw[-] (15,0) -- (15,4);

      	                                        \draw (0,-.8) node(c){\small $0$};
      	                             \draw (20/2,-1.2) node(c){\small $\frac{1}{2}$};
									 
									  \draw (20/4,-1.2) node(c){\small $\frac{1}{4}$};

      	 \draw (20/2,-2.5) node(c){\small Time};

      	                             \draw (20,-1) node(c){\small$1$};

      	\draw (15,-1.2) node(c){\small$\frac{3}{4}$};

      	    \draw(-3,1) node(z){\small Agent $2$};
      	                 \draw(-3,3) node(z){\small Agent $1$};

   	      	\draw(2.5,3) node(z){\small $a$};

   	      	\draw(2.5,1) node(z){\small $a$};
			
   	      	\draw(10,3) node(z){\small $b$};

   	      	\draw(10,1) node(z){\small $a$};

   	      	\draw(17.5,3) node(z){\small $b$};

   	      	\draw(17.5,1) node(z){\small $b$};

      	  \end{tikzpicture}
      	 
      	 \caption{Illustration of VER for the problem in Example~\ref{example:nonbi}.}
      	\label{figure:ver}
      	\end{center}
		
		\end{figure}
\end{itemize}

\end{example}

We first verify that each step of VER is well-defined and that it always terminates. 
It is then clear from the definition that it results in a feasible allocation.
The proof proceeds by showing that agents consume equivalence classes in decreasing order of their preferences and then observing that in each round some agent moves to a less preferred class or becomes inactive since she cannot increase any of her guarantees further.
(All proofs are in the appendix.)

\begin{restatable}{theorem}{verterminates}\label{thm:verterminates}
	For every closed set $X$ of allocations, \Cref{algo:closedver} terminates after at most $m\cdot n$ rounds with an allocation $p\in X$.
\end{restatable}
	
Intuitively, it is clear that VER yields an allocation that is \dl-efficient (and, thus, \sd-efficient) among feasible allocations. 
To see this, assume that the allocation $p$ of VER were dominated by some allocation $q$ in $X$.
We can view $q$ as the result of some eating trajectory different from the one which yields $p$ in \Cref{algo:closedver}.
Assuming agents eat objects at unit rate in decreasing order of their preference pins down a unique trajectory for $q$.
Consider the earliest point in time at which both trajectories diverge and let $i^*$ be the lexicographically smallest agent who starts eating different objects (or, more generally, equivalence classes) in both trajectories.
Since $q$ dominates $p$, $i^*$ strictly prefers her object in the $q$-trajectory to that in the $p$-trajectory.
But this contradicts the choice of $o_{i^*}$ at that time in Line~\ref{ln:oi} of \Cref{algo:closedver}.
	
\begin{restatable}{theorem}{constrainedefficient}\label{thm:constrainedefficient}
	When VER is applied to any closed set of allocations $X$, it returns an allocation that is constrained \dl-efficient among the allocations in $X$.
\end{restatable}

The steps in \Cref{algo:closedver} that are potentially computationally taxing are determining the best objects an agent can eat at a given time in Line~\ref{ln:oi} and the amount of time the current eating scheme can be continued in Line~\ref{ln:delta}. 
Both require maximizing a linear objective over the intersection of $X$ with a polytope.
Whenever $X$ is itself a polytope (described by polynomially many linear inequalities), this can be done efficiently by solving a linear program. 
More generally, if $X$ is a union of polytopes, we can solve the problem for each polytope individually and then pick the best solution (which is the best object $o_i$ agent $i$ can eat or the maximal amount of time $\delta^*$ the current eating scheme can continue).

\begin{restatable}{theorem}{polynomial}\label{thm:polynomial}
	If $X$ is a finite union of polytopes, the outcome of VER can be computed in polynomial time in the number of inequalities used to describe $X$.
\end{restatable}

\subsection{Ordinal Egalitarianism and Equal Treatment of Equals}

\citet{Bogo15b} proposed a solution notion called \emph{ordinal egalitarianism} that characterizes the probabilistic serial rule (which can be argued to be the fairest rule for the allocation of objects to agents under ordinal preferences).  
Ordinal egalitarianism is defined as follows.
We evaluate any given allocation by the list of numbers $t_i^j$, the total probability agent $i$ gets of objects from her first $j$ equivalence classes, for all $i$ and $j$. 
An allocation is \emph{ordinally egalitarian} if it is a leximin maximizer of the signature vector $t^\uparrow(p) = (t_i^j(p))_{i,j}$ over all feasible allocations. 
For any signature vector $t^\uparrow(p)$, we will denote by $t^\uparrow(p)(j)$ as the $j$-th entry of the non-decreasing ordered entries. 
We will denote by $t_i(p)$ as the signature vector $t^\uparrow(p)$ in which only the entries corresponding to agent $i$ are considered.
If an allocation is ordinally egalitarian, then it is \dl-efficient.
We show that under convex constraints, VER returns an allocation that is ordinally egalitarian.

\begin{restatable}{theorem}{ordinalegal}\label{thm:ordinalegal}
	The VER allocation for a closed and convex set of allocations $X$ is ordinally egalitarian among the allocations in $X$. 
\end{restatable}

\begin{remark}\label{rem:convex}
	It follows from \Cref{thm:ordinalegal} that if $X$ is convex, the outcome of VER is independent of the priority order over agents, which determines in which order agents choose objects.
	Hence, in that case, like the eating-based algorithms in the literature, VER can be formulated so that it is symmetric with respect to the agents. 
\end{remark}

One of the desiderata when assigning objects to agents is to treat all agents equally.
Roughly, this means that if two agents have the same preferences and are subject to the same constraints, then they should receive the same allocation.
Thus, there are two ways in which asymmetry between agents in an allocation can arise: the mechanism that determines it is inherently asymmetric or the set of feasible allocations is biased toward certain agents.

VER introduces some asymmetry of the first kind by letting agents choose their most-preferred available object in lexicographic order. 
In general, this order is relevant since whether or not an object is available to an agent may depend on the choices of previous agents.
It is clear that this asymmetry cannot be eliminated entirely.
For example, two agents may have the same most-preferred object and only the allocations assigning that object to one of them with probability 1 may be feasible.
(This is just the special case of deterministic mechanisms.)
However, if $X$ is convex, the order in which agents choose objects is irrelevant~(see \Cref{rem:convex}).
A second source of asymmetry is the set of feasible allocations $X$ itself.
Clearly, if $X$ is asymmetric, we may not be able to give the same allocation to agents with the same preferences.

These considerations motivate the conditions under which VER treats agents equally.
Let $i$ and $i'$ be two agents.
We say that a set of allocations $X$ is symmetric for $\{i,j\}$ if $p\in X$ implies $p^{(ij)}\in X$, where $p^{(ij)}$ is identical to $p$ except that the allocations of $i$ and $j$ are swapped.
Moreover, $X$ is convex for $\{i,i'\}$ if for all $\lambda\in[0,1]$, $p\in X$ and $p^{(ii')}\in X$ implies $\lambda p + (1-\lambda) p^{(ii')}\in X$.

\begin{restatable}{theorem}{symmetric}\label{thm:symmetric}
	Let $X$ be a closed set of allocations and $\apref$ be a preference profile.
	If $X$ is symmetric and convex for all $i,i'\in N$ with ${\pref_i} = {\pref'_i}$, then the outcome of VER satisfies equal treatment of equals.
\end{restatable}

\subsection{Variable Eating Rates}

	One can generalize VER by equipping each agent $i$ with an eating rate function $f_i\colon\mathbb R_+\rightarrow\mathbb R_+$ (assumed to be measurable), which determines how fast she can increase her probability for an equivalence class as a function of time.
	\Cref{algo:closedver} can be extended to eating rate functions by introducing variables $t^k$ with $t^0 = 0$ and $t^k = t^{k-1} + \delta^*$ as determined in Line~\ref{ln:delta} in round $k$, replacing the second line in the linear program in Line~\ref{ln:delta} by
	\begin{align*}
		p(i,E_i)&\ge  \pi^k(i,E_i) + \int_{t_k}^{t_k + \delta} f_i(s)ds \qquad\forall i\in N',
	\end{align*}
	and replacing Line~\ref{ln:pk+1a} by 
	\begin{align*}
		\pi^{k+1}(i,E_i)\longleftarrow \pi^{k}(i,E_i) + \int_{t_k}^{t_k + \delta^*} f_i(s)ds \qquad\forall i\in N'.
	\end{align*}
	If the eating rate functions are piecewise constant, which is the case we will consider, $\delta^*$ can still be computed by linear programming.
	In general, this is not possible, however.
	
	\citet[][Theorem 1]{BoMo01a} show that for unit demand object allocation without constraints, an allocation is \sd-efficient if and only if it is the outcome of their eating mechanism for some profile of eating rate functions.
	The ``if part'' remains true in our setting, that is, for every profile of eating rate functions, VER yields a constrained \dl-efficient and, thus, \sd-efficient allocation.
But not every \sd-efficient allocation is the outcome of VER for some profile of eating rate functions.
	Consider the following preferences.
	\begin{center}
	\begin{tabular}{lccccc}
	$\succ_1$:&$a$&$b$&$c$\\
	$\succ_2$:&$b$&$c$&$a$\\
	$\succ_3$:&$c$&$a$&$b$\\
	\end{tabular}
	\end{center}
	Let $X$ be the convex hull of the allocations $p$ and $q$.
	\begin{center}
		$p=$
					 \begin{blockarray}{ccccccccc}
					 		&&\matindex{$a$}&\matindex{$b$}& \matindex{$c$}&\\
					 	    \begin{block}{c(cccccccc)}
					 			\matindex{$1$}& &$\nicefrac12$&$0$&$\nicefrac12$& \\
					 			\matindex{$2$}& &$\nicefrac12$&$\nicefrac12$&$0$& \\
								\matindex{$3$}& &$0$&$\nicefrac12$&$\nicefrac12$& \\
					 	    \end{block}
					 	  \end{blockarray}
		\qquad\qquad
		$q=$
					 \begin{blockarray}{ccccccccc}
					 		&&\matindex{$a$}&\matindex{$b$}& \matindex{$c$}&\\
					 	    \begin{block}{c(cccccccc)}
					 			\matindex{$1$}& &$\nicefrac13$&$\nicefrac13$&$\nicefrac13$& \\
					 			\matindex{$2$}& &$\nicefrac13$&$\nicefrac13$&$\nicefrac13$& \\
								\matindex{$3$}& &$\nicefrac13$&$\nicefrac13$&$\nicefrac13$& \\
					 	    \end{block}
					 	  \end{blockarray}
						  \end{center}
		Every allocation in $X$ is constrained \sd-efficient. 
		However, there is no profile of eating rate functions for which VER will yield $p$.
		This is because at each point during the eating process, every agent's favorite object $o$ will be available to her until she has consumed $\frac12$ of $o$.
		
	The reason for the equivalence derived by \citeauthor{BoMo01a} is that \sd-efficiency and \dl-efficiency coincide in their setting as shown by \citet{ChDo16a}.
	This is no longer true in our case as observed above.
We can however show that when the set of feasible allocations is convex, every constrained \dl-efficient allocation is unanimously indifferent to the outcome of VER for some profile of eating rate functions. 
To this end, it suffices to consider indicator functions as eating rate functions.
Those result in one-at-a-time-VER, which at any point in time, allows only one agent to increase her probability for some equivalence class. 
Formally, for all $i\in N$, $f_i = \chi_{S_i}$, where the $S_i$ are pairwise disjoint measurable subsets of $\mathbb R_+$ and $\chi_S$ is the indicator function of $S\subset\mathbb R_+$.
Clearly, the outcome of VER can be achieved by some instance of one-at-a-time-VER.

 \begin{restatable}{theorem}{dlthm}\label{thm:dl}
	 Let $X$ be a closed and convex set of allocations and $p \in X$.
 	 Then $p$ is constrained \dl-efficient with respect to $X$ if and only if all agents are indifferent between $p$ and the outcome of some instance of one-at-a-time-VER on the set $X$.
  \end{restatable}
  
  \begin{remark}\label{rem:oneatatimever}
   	An interesting instance of one-at-a-time-VER is the \emph{vigilant priority rule}. 
  Fix some order of the agents. 
  The first agent in that order increases her probability guarantees for each of her equivalence classes in order of her preferences by as much as possible.
  Once no further increase is possible for the first agent, it is the second agent's turn to increase all of her guarantees in order of her preferences by as much as possible.
  We proceed in the same way for the remaining agents.
  \Cref{app:vp} gives a pseudo-code formulation of this algorithm.\footnote{We obtain the vigilant random priority rule as an instance of one-at-a-time VER by setting $S_i = [(l-1)|O|,l|O|)$, where $l$ is the rank of agent $i$ in the chosen order.} 
  The vigilant priority rule is clearly unfair to agents who come later in the order. 
  It is (strongly) \sd-strategyproof whenever the set of feasible allocations is independent of the agents' preferences. 
  If we choose the order of agents uniformly at random, then the rule is no longer \dl-efficient.\footnote{To guarantee that the convex combination of the outcomes of the vigilant priority rule for different orders is feasible, we need to assume that the set of feasible allocations is convex.}
  The reason is that it coincides with the classic random priority rule in the assignment domain, which is well-known to violate \sd-efficiency (and thus \dl-efficiency).
  \end{remark}

\subsection{Examples of Constraints}

We conclude this section by giving concrete classes of constraints that VER can handle.

\paragraph{Distributional constraints on the ex-ante allocations}

We have pointed out that VER applies to a wide variety of constraints on the ex-ante allocations. 
A natural and general family of constraints arises from imposing lower and upper bounds on the cumulative probability a subset of agents can obtain from a subset of objects. 
These constraints appear in applications like school choice if a certain number of seats is reserved for students who live close to the school.  
Another class of constraints obtains from diversity requirements, which put bounds on the fraction of objects from within some subset that is assigned to a certain subset of agents (instead of bounding the absolute number).  

\paragraph{Distributional constraints on the ex-post allocations}
As mentioned in the introduction, ex-post constraints, that is, constraints on deterministic allocations, can be expressed as ex-ante constraints on the set of probabilistic allocations by taking only those allocations to be feasible which can be decomposed into deterministic allocations satisfying the ex-post constraints.
These can capture upper and lower bounds on the number of objects a subset of agents can receive from some subset of objects ex-post, for example.
One instance of this are bi-hierarchical constraints on deterministic allocations of the type studied by \citet{BCKM12a}.
Similarly, one can also formulate diversity constraints ex-post.

Ex-post constraints can be combined with possibly different ex-ante constraints on probabilistic allocations. 
For example, one may stipulate that the ex-ante constraints hold exactly, while the ex-post constraints only need to be satisfied approximately.
\citet{AkNi20a} show that if the ex-ante and ex-post constraints are the same, then any allocation satisfying the ex-ante constraints can be decomposed into deterministic allocations that approximately achieve the constraints.
Hence, in that case, there is little need to impose additional constraints which guarantee that a decomposition into deterministic allocations satisfying the constraints exists.

\paragraph{Allocation with endowments}
VER can be applied to problems in which agents have endowments to obtain individually rational and possibly more efficient re-allocations. 
\citet{AtSe11a} presented the controlled consuming algorithm, which can be viewed as an extension of the probabilistic serial rule.
It ensures that each agent gets an allocation she weakly prefers to her endowment with respect to stochastic dominance. 
The requirement of \sd-individual rationality can be embedded in our framework by using a linear number of inequalities for each agent.  
Hence, we can use linear programming to capture the controlled consuming algorithm. 
In principle, we can even combine individual rationality constraints and priorities.

\paragraph{Deterministic and integral allocations}
Our approach has been framed in the context of randomization or time-sharing.
However, it also applies when no randomization is involved or allowed by rendering all non-deterministic allocations infeasible. 
Then the $\delta$ increment in the probability of an agent for an object (cf. \Cref{algo:closedver}, Line~\ref{ln:delta}) always has to be one. 
This ensures that an agent either gets an object with probability one or zero. 
Not allowing probabilistic allocations under feasibility constraints may however render many problems NP-hard and thus computationally intractable. 
We note that for the house allocation problem with deterministic allocations~\citep{Sven99a}, VER coincides with the serial dictatorship rule.

\paragraph{Welfare requirements}

If the agents additionally have cardinal preferences over the objects, then minimum social welfare guarantees can also be treated as feasibility constraints.

\section{Stable Probabilistic Allocations}\label{sec:stable}

One way in which feasibility constraints can occur is if agents have different priorities for objects and we require allocations to be stable.
How to generalize stability to probabilistic allocations is not at all unambiguous, however.
Various notions have been proposed in the literature, of which we will discuss four.
Then we examine VER as a mechanism for stable probabilistic object allocation under priorities.
Many of its properties follow directly from the general statements we have derived above.
On the other hand, VER lacks properties such as strategyproofness, unconstrained efficiency, and envy-freeness.
We show that this is unavoidable if one insists on stability.

We augment our formal model by assuming that every object $o$ comes with a (complete, reflexive, and transitive) relation $\pref_o$ over agents, which specifies the agents' priorities for $o$.
A profile $\p = (({\pref_i})_{i\in N},({\pref_o})_{o\in O})$ consists of preferences for each agent and priorities for each object.
The priorities give rise to a relaxed notion of equal treatment of equals, which only requires that two agents receive the same allotment if they have the same priority for all objects.
That is, an allocation $p$ satisfies \emph{limited equal treatment of equals} for a profile $\p$ if $p(i) = p(j)$ for all $i,j\in N$ with ${\pref_i} = {\pref_j}$, $\capa(i) = \capa(j)$, and $i\sim_o j$ for all $o\in O$.
The use of \emph{preferences} for the agents and \emph{priorities} for objects follows the school choice literature~\citep{AbSo03b}, which considers efficiency, fairness, and strategic aspects only from the perspective of the agents. 
It also makes the connection with problems where objects have no priorities clearer.

\subsection{Notions of Stability}\label{sec:stabilitynotions}

We consider four notions of stability, all of which coincide with the standard (pairwise) stability for deterministic allocations if the agents' preferences over allocations are responsive. 
Throughout this section, we assume that every agent $i$ has a capacity $\capa(i)\in\mathbb N$ of objects she can receive.
We stipulate that for every allocation $p$, $p(i,O) \le \capa(i)$.
Recall that a deterministic allocation $p$ is \emph{stable} if for all $i\in N$ and $o\in O$, either $\sum_{o'\succsim_i o} p(i,o') = \capa(i)$ or $p(o)\succsim_o i$.

Our most restrictive notion is ex-ante stability, which has been introduced by \citet{KeUn15a}.
It prescribes that an agent $i$ can only receive a positive probability for an object $o$ if every agent $j$ with a higher priority for $o$ can meet the capacity $\capa(j)$ with objects $j$ prefers to $o$.
Formally, $p$ is ex-ante stable if for all $i,j\in N$ and $o\in O$,
\begin{align}
	j\succ_o i \text{ and } p(i,o)>0 \text{ implies } \sum_{o'\succsim_j o} p(j,o') = \capa(j)\tag{ex-ante stability}
\end{align}
Thus, ex-ante stability requires that $j$ has no justified envy toward $i$ even before knowing the realization of the random allocation $p$.

Analogously, one can ask that no agent should have justified envy ex-post, that is, after a deterministic allocation has been selected according to the random allocation.
This leads us to define $p$ as ex-post stable if
\begin{align}
	p \text{ is a convex combination of deterministic stable allocations}
	\tag{ex-post stability}
\end{align}

The third stability notion, called fractional stability, requires that for all $i\in N$ and $o\in N$,
\begin{align}
	\sum_{o'\succsim_i o,\, o'\neq o} p(i,o') + \capa(i)\,\sum_{j\succsim_o i} p(j,o) \ge \capa(i)\tag{fractional stability}\label{eq:fs}
\end{align}
These inequalities originate from the work of \citet{RRV93a}, who observed that for deterministic allocations, their conjunction is equivalent to stability. 
\citet{BaBa00b} showed that when preferences and priorities are strict, fractional stability is equivalent to ex-post stability.
\citet{AzKl19a} give an example, attributed to Battal Do\u{g}an, which shows that this equivalence breaks down if one omits the strictness assumption for both preferences and priorities.
Our \Cref{ex:fractionalnotex-post} in the appendix shows that weak priorities alone to break the equivalence.

A motivation for fractional stability under unit capacities, adopted from \citet{AzKl19a}, is that if the inequality for a pair $(i,o)$ fails, then $i$ justifiably envies the set of agents with lower priority for $o$ for jointly consuming more of $o$ than $i$ consumes of objects $i$ weakly prefers to $o$.
Another reason for considering fractional stability is as a proxy for ex-post stability in situations where it is computationally prohibitive to handle the latter.
Since the set of fractionally stable allocations is described by $m \cdot n$ linear inequalities, it is typically much more well-behaved in that respect.

Lastly, we consider claimwise stability, which has been introduced by \citet{Afac18a}.
We say that agent $i$ has a justified claim against $j$ for object $o$ if $i$ has higher priority for $o$ than $j$ and $j$'s probability for $o$ is larger than $i$'s probability for objects $i$ weakly prefers to $o$.
An allocation $p$ is claimwise stable if no agent has a justified claim, that is, if for all $i,j\in N$ and $o\in O$,
\begin{align*}
	 i\succ_o j \text{ implies } \sum_{o'\pref_i o,\, o'\neq o} p(i,o') \ge p(j,o)\tag{claimwise stability}
\end{align*}
\citet{AzKl19a} showed that each of our four stability notions implies the ones below it in the list, while none of the converse implications holds.

\subsection{Vigilant Eating on Sets of Stable Allocations}\label{sec:verstable}

In this section, we study the properties of VER when it is applied to sets of stable allocations.
If $S$ is one of our stability notions, then $S$-VER denotes the mechanism which, for a profile $\p$, runs VER for the preferences $\apref$ on the set $X$ of $S$-stable allocations for $\p$, so that $S\text-\ver\p = \ver(\apref,X)$.
Most of the properties of VER on sets of stable allocations follow directly from the results we have proved in \Cref{sec:ver}.
In particular, for every stability notion $S$ defined in \Cref{sec:stabilitynotions}, $S$-VER yields an allocation that is $S$-stable, constrained \sd-efficient, and satisfies limited equal treatment of equals.
We summarize these results in the following corollary.

\begin{restatable}{corollary}{stableandefficient}\label{cor:stableandefficient}
	For each of our four stability notions $S$, $S$-VER returns an $S$-stable allocation that is $\sd$-efficient among $S$-stable allocations and satisfies limited equal treatment of equals.
\end{restatable}

For the extreme cases of coarseness of priorities and strict preferences, VER coincides with well-known mechanisms. 
If priorities are flat, that is, if all agents have the same priority for every object, every allocation is stable for each of our stability notions. 
In that case, VER reduces to the probabilistic serial rule of \citet{BoMo01a}, which corresponds to unconstrained eating.
The opposite extreme is that priorities are strict. 
Then VER returns the agent optimal deterministic stable allocation, which is also the outcome of the agent-proposing deferred acceptance algorithm. 
Intuitively, this checks out since VER is optimal for agents in the sense that it allows them to eat their most preferred object available to them.

\begin{restatable}{corollary}{agentoptimal}\label{cor:agentoptimal}
	Assume that preferences and priorities are strict.
	Then, for each of our four stability notions $S$, $S$-VER returns the agent optimal deterministic stable allocation.
\end{restatable}

In view of \Cref{cor:stableandefficient} and \Cref{cor:agentoptimal}, we also recover the well-known result that the agent proposing deferred acceptance algorithm returns an allocation that is Pareto efficient within the set of stable allocations when the preferences and priorities are strict. 
\Cref{cor:agentoptimal} shows that $S$-VER results in the same mechanism for all four stability notions in that case even though the sets of stable allocations are not the same (except for ex-post stability and fractional stability).
For weak priorities, all four instantiations of VER are in fact distinct. 
We provide examples in \Cref{sec:proofssection5}.

An alternative interpretation of our formal model is that the entities on both sides of the market are agents who have preferences over the other side (instead of one side being objects with priorities over agents).
Then, instead of considering \sd-efficiency for one side, we could ask for allocations that are \sd-efficient with respect to the preferences of both sides, called two-sided \sd-efficiency.
Formally, an allocation $p$ is \emph{two-sided \sd-efficient} for a profile $\p$ if there is no allocation $q$ such that $q(i)\pref_i^\sd p(i)$ for all $i\in N$ and $q(o)\pref_o p(o)$ for all $o\in O$ and at least one preference is strict.
For ex-ante stability and fractional stability, we can show that VER always returns an allocation that is two-sided \sd-efficient among all allocations, not only among stable allocations.
For ex-post and claimwise stability, this is open.

\begin{restatable}{proposition}{stablebothsides}\label{prop:sstablebothsides}
	Assume that preferences are strict.
	For ex-ante stability and fractional stability, $S$-VER returns an allocation that is two-sided \sd-efficient.
\end{restatable}

Let us now consider the computational complexity of VER for sets of stable allocations.					
It is easy to see from the definitions that the set of fractionally stable allocations and the set of claimwise stable allocations are polytopes described by on the order of $m\cdot n$ and $m\cdot n^2$ linear inequalities, respectively. 
Thus,	\Cref{thm:polynomial} implies that the corresponding vigilant eating rules can be computed in polynomial time.

\begin{corollary}\label{cor:polynomial}
	FS-VER and CWS-VER can be computed in polynomial time.
\end{corollary}

The exact complexity of computing VER on the set of ex-post stable and ex-ante stable allocations is not settled. 
Since the set of ex-post stable allocations is convex, it would by \Cref{thm:polynomial} suffice to describe it by a polynomial number of linear inequalities. 
For the case of strict preferences and priorities, this has been done by \citet{BaBa00b} as discussed earlier. 
For weak priorities, it is an interesting open problem, as also pointed out by \citet{KeUn15a}.
Ex-ante stability can be captured by a set of constraints, each of which is a disjunction of two linear equalities: for each $i\in N$ and $o\in O$, either $\sum_{o'\pref_i o}p(i,o)=1$ or $\sum_{j\prec_o i}p(j,o)=0$. 
In general, solving problems involving disjunctions of equalities is NP-hard, however.

\subsection{Incompatibility of Stability with Efficiency, Envy-Freeness, and Strategyproofness}

For each of our stability notions, VER on the set of stable allocations violates unconstrained efficiency, weak envy-freeness, and weak strategyproofness when preferences over probabilistic allocations are based on stochastic dominance.
Neither of those shortcomings is specific to VER, but the consequence of an inherent incompatibility of each of these properties with stability and constrained efficiency.
We address them in turn.

\subsubsection*{Efficiency}
\citet{Roth82a} showed that there may be no deterministic allocation that is both stable and Pareto efficient.
Since any ex-post stable and \sd-efficient allocation has to be a convex combination of stable and Pareto efficient, it follows that ex-post stability is incompatible with unconstrained \sd-efficiency.
We prove that this conflict remains even when weakening stability to claimwise stability.

\begin{restatable}{proposition}{sdclaimwisedisjoint}\label{prop:efficientandstable}
	There may be no allocation that is claimwise stable and \sd-efficient.
\end{restatable}
	Hence, for any stability notion that is stronger than claimwise stability, the set of stable allocations can be disjoint from the set of \sd-efficient allocations.

\subsubsection*{Envy-Freeness}
In the absence of priorities, envy-freeness requires that each agent prefers her allocation to that of any other agent. 
This definition is no longer compelling for non-trivial priorities since an agent may legitimately receive an allocation that some other agent would prefer to her own because of a higher priority for some objects.
One definition of fairness in this context is limited equal treatment of equals, which we discussed above.
Another one is that of justified envy.
It applies only to pairs of agents $i,j$ who have the same priority for all objects and requires that $i$ does not prefer $j$'s allotment to her own when comparing them via stochastic dominance.
That is, an allocation $p$ is weakly \sd-envy-free if for all $i,j\in N$,
\begin{align*}
	i\sim_o j\text{ for all $o\in O$ implies } p(j)\not\succ_i^{\sd} p(i)\tag{weak \sd-envy-freeness}
\end{align*}
Even this weak notion of envy-freeness turns out to be incompatible with claimwise stability and hence with all stability notions stronger than that.
This can be seen from the example in the proof of \Cref{prop:efficientandstable}.

\begin{proposition}\label{prop:cwsenvyfree}
		There may be no allocation that is claimwise stable and weakly \sd-envy-free.
\end{proposition}

\subsubsection*{Strategyproofness}		

Comparing allocations via stochastic dominance results in incomplete preferences.
Thus, there are two notions of strategyproofness associated with it.
The stronger, usually called \sd-strategyproofness, requires that the allotment obtained by truth-telling weakly stochastically dominates any allotment that can be obtained otherwise.
\citet{BoMo01a} proved that when priorities are flat, there exists no mechanism that is \sd-strategyproof, \sd-efficient, and satisfies equal treatment of equals. 
Under flat priorities, \sd-efficiency is the same as constrained efficiency for each of our stability notions since stability has no bite.
We can thus not hope for a mechanism that returns a stable allocation and satisfies \sd-strategyproofness, constrained efficiency, and limited equal treatment of equals, irrespective of which stability notion we choose.

The weak notion of \sd-strategyproofness only prescribes that no agent can obtain a strictly \sd-dominating allotment by misrepresenting her preferences.
Formally, a mechanism $f$ is weakly \sd-strategyproof if for all agents $i\in N$ and all profiles $\p,(\apref['],{\succsim}_O)$ with ${\pref_j} = {\pref'_j}$ for all $j\in N\setminus\{i\}$,
\begin{align*}
 f(\apref['],{\succsim}_O)(i)\not\succ_i^{\sd} f\p(i).
\end{align*}
Now for flat priorities and strict preferences, there is a mechanism that satisfies weak \sd-strategyproofness, \sd-efficiency, and equal treatment of equals: the probabilistic serial rule.
Moreover, the generalized probabilistic serial rule for upper quotas forming a bi-hierarchy of \citet{BCKM12a} is weakly strategyproof.
But imposing stability again results in an impossibility. 
\citet{Afac18a} proved that for at least 4 agents, there exists no mechanism that is weakly \sd-strategyproof and always returns a claimwise stable and constrained \sd-efficient allocation.
We extend his result to the remaining three stability notions.
That is, for any stability notion in our list, no mechanism is jointly weakly \sd-strategyproof, stable, and constrained efficient (with respect to the set of stable allocations).
Note that constrained efficiency becomes weaker as stability becomes stronger, since the set of stable allocations becomes smaller.
Thus, none of these statements implies the other.
Our proof does not rely on weak preferences and requires only 3 agents.
		
\begin{restatable}{proposition}{strategyproof}\label{prop:strategyproof}
	No mechanism satisfies weak \sd-strategyproofness, S-stability, and S-constrained efficiency for $\text{S}\in\{\text{ex-ante, ex-post, FS}\}$ even if preferences are strict. 
\end{restatable}

All the properties are needed for the conclusion of \Cref{prop:strategyproof}.
VER satisfies S-stability and S-constrained efficiency.
Ignoring the priorities and running the probabilistic serial rule gives a mechanism that is weakly \sd-strategyproof and efficient (and thus S-constrained efficient), but not S-stable.
The deferred acceptance algorithm with lexicographic tie-breaking of priorities is \sd-strategyproof and S-stable. 
For ex-post stability, fractional stability, and claimwise stability, one can simultaneously achieve limited equal treatment of equals by breaking ties uniformly at random.

\subsection{Comparison of Probabilistic Allocation Mechanisms}\label{sec:otherrules}

We discuss other mechanisms presented in the literature.
Since most of them have only been defined for agents with strict preferences, we assume that preferences are strict for the comparison.
We have already discussed the probabilistic serial rule in the introduction. 
Except for the first mechanism (random priority), all the other mechanism
are extensions of the deferred acceptance algorithm, which is typically defined for strict preferences and strict priorities~\citep{GaSh62a,Roth08a} and returns a stable allocation.

\paragraph{Random Priority} (or random serial dictatorship) chooses an ordering of the agents uniformly at random and then lets each agent pick her most preferred object among the ones remaining in that order.~\citep{BoMo01a,ABB13b}. 
The rule does not take into account the priorities of the objects. 
For the basic assignment problem, random priority is known to be strategyproof. It satisfies equal treatment of equals but is not \sd-efficient or \sd-envy-free~\citep{BoMo01a}. If the priority order is chosen deterministically, then \sd-efficiency is regained, but equal treatment of equals is lost.
	
\paragraph{Deferred acceptance with lexicographic tie-breaking} 
 A simple adaptation of the deferred acceptance algorithm to the case of weak priorities is to break the ties and then run deferred acceptance.
If the tie-breaking is pre-determined, for example lexicographically over agents, and thus does not depend on the agents' preferences, the resulting mechanism is strategyproof.
Moreover, it returns a deterministic allocation that is stable and hence ex-ante stable.
However, it is not \sd-efficient even among deterministic stable allocations. 
Like any other mechanism that returns deterministic allocations, it also violates limited equal treatment of equals.

\paragraph{Deferred acceptance with random tie-breaking}

Another natural approach is to break ties in the priorities uniformly at random and then run deferred acceptance. 
\citet{Afac18a} referred to this extension as \emph{probabilistic deferred acceptance}. 
To compute the random allocation, we need to take the mean of the outcomes over all possible tie-breakings. 
Under flat priorities, probabilistic deferred acceptance is equivalent to random priority. 
Since the latter is well-known to violate \sd-efficiency, it follows that probabilistic deferred acceptance does not satisfy constrained \sd-efficiency. 

\paragraph{Fractional deferred acceptance and trading} \citet{KeUn15a} presented the fractional deferred acceptance and trading mechanism.
It returns an allocation that is strongly ex-ante stable and \sd-efficient constrained to the set of strongly ex-ante stable allocations. 
It is different from the probabilistic serial rule under flat priorities and thus also from VER.

\paragraph{Constrained probabilistic serial} 
\citet{Afac18a} introduced the \emph{constrained probabilistic serial} rule, which is an adaptation of the probabilistic serial rule that obtains a claimwise stable allocation.
Because of the simplicity of the constraints imposed by claimwise stability, this algorithm does not need to look ahead to check when an agent should stop eating an object. 
Suppose an agent $i$ starts eating object $o$. 
At that point, we put an upper bound on all agents $j$ who have a lower priority for $o$ than $i$.
If $j$ had been eating $o$ all the time while $i$ was eating more preferred objects, we can stop $j$ from eating more. 
If $j$ had been eating $o$ only part of the time, we put a limit on how much $j$ can eat $o$. 
The upper limit is equal to the amount of time $i$ was eating objects weakly more preferred objects. 
The constrained probabilistic serial rule can be viewed as a careful version of the probabilistic serial rule that handles stability constraints dynamically.
By contrast, VER for claimwise stability makes look-ahead checks to see how much of an object an agent can eat before a stability violation becomes unavoidable.
One can show that in this case, the look-ahead checks are not necessary and both mechanisms coincide.

Table~\ref{table:summary} summarizes the properties satisfied by various mechanisms. 

\begin{table*}[tb]
\centering
\makebox[\textwidth][c]{
\begin{adjustbox}{max width=1.1\textwidth}
\begin{tabular}{lcccccccccccccccc}
\toprule
	& \multirow{2}{*}{$S$-VER} & probabilistic & random & deferred acceptance &deferred acceptance & fractional deferred   \\
	&  & serial & priority & (lexicographic) & (uniform) &acceptance and trading \\
	\midrule
	$S$-stability & + & -- & -- & +\footnotemark & + & +\\
	$S$-constrained \sd-efficiency & + & + & -- & -- &--& --\footnotemark\\
	\sd-efficiency & -- & + & -- & -- & -- & --\\
	\sd-envy freeness & -- & + & -- & -- & -- & -- \\
	weak \sd-envy freeness & -- & + & + & -- & -- & -- \\
limited equal treatment of equals & + & + & + & + & -- & + \\
  \bottomrule
\end{tabular}% 
\end{adjustbox}
}
\caption{Summary of the properties satisfied by the allocation mechanisms discussed in \Cref{sec:otherrules}. $S$ may stand for each of the four stability notions we define in \Cref{sec:stabilitynotions}. In order to enable a comparison of all rules, we assume that preferences are strict.}
\label{table:summary}
\end{table*}

\addtocounter{footnote}{-2}
\stepcounter{footnote}\footnotetext{The deferred acceptance algorithm with lexicographic tie-breaking satisfies ex-post stability and all weaker stability notions. 
It violates ex-ante stability, however.}
\stepcounter{footnote}\footnotetext{The fractional deferred acceptance and trading algorithm violates $S$-constrained \sd-efficiency when $S$ is ex-post stability or any weaker stability notion. It satisfies \sd-efficiency constrained to the set of ex-ante stable allocations, however.}

\section{Extension Beyond Allocation Problems}

The VER framework is not restricted to allocation problems. 
It applies just as well to other settings where fractional or probabilistic outcomes are feasible, such as social choice \citep{BMS05a,Bran17a}, coalition formation~\citep{BoJa02a,ABH11c}, networks~\citep{JaWo96a}, and other models discussed by \citet{Sonm99a}.
In this section, we discuss how to extend our model to capture these applications.

Instead of a set of objects, we now consider an abstract set of alternatives $A$ and each agent has a preference relation over $A$.
An outcome is an element of $\mathbb R_+^A$. 
The agents' preferences over outcomes are again based on stochastic dominance.
By a problem we denote a pair $(\apref,X)$, where $\apref$ is a preference profile and $X\subset \mathbb R_+^A$ is a non-empty set of feasible outcomes.
As before, a mechanism $f$ maps a problem $(\apref,X)$ to an outcome $f(\apref,X)\in X$.

A generalization of VER then works as follows: in each round, agents are addressed in lexicographic order as in \Cref{algo:closedver}. 
When it is agent $i$'s turn, we determine $i$'s most preferred alternative $a_i$ whose probability can still be increased (while also increasing the probability for alternatives $a_j$ with $j < i$). 
Having determined $a_i$ for every agent $i$, we find the maximal $\delta$ so that the probability for all $a_i$ can be increased by at least $\delta$.
Note that the probability of every alternative $a_i$ is increased by the same amount independently of how many agents nominate $a_i$.

This version of VER can address several classes of problems.
	\paragraph{Probabilistic voting}
	Our abstract model immediately captures probabilistic voting \citep[see, for example,][]{Bran17b} where voters have preferences over the alternatives $A$ and an outcome is a probability distribution over the alternatives.
	That is, $X = \Delta(A) \subset \mathbb R_+^A$.
	In this context, VER coincides with the ESR rule of \citet{AzSt14a}. 
	\paragraph{Participatory budgeting}
	The probabilistic voting setting can also be interpreted as determining the share of the budget allocated to each of the alternatives~\citep{AAC+19a}. Since our model allows for arbitrary constraints, it can capture natural constraints such as enforcing lower bounds (reflecting the minimum funding required) on alternatives that get at least some funding. 
	\paragraph{Probabilistic allocation with bundles}
	Let $O$ be a set of objects and assume every agent $i$ has a preference relation $\hat\pref_i$ over \emph{subsets} of $O$ \citep[see, for example,][]{ChLi20a}.
	A deterministic allocation of objects to agents is an ordered $N$-partition of objects (which may include empty sets).
	Let $A$ be the set of all such partitions. 
	By $a(i)$ we denote the set of objects agent $i$ receives in the allocation $a$.
	The preference relation $\pref_i$ over $A$ has $a\pref_i b$ if and only if $a(i) \mathrel{\hat\pref_i} b(i)$ for all $a,b\in A$.
	A random allocation is an element of the unit simplex $\Delta(A) \subset \mathbb R_+^A$ and $X\subset \Delta(A)$ specifies a set of feasible probabilistic allocations.
	\paragraph{Two-sided probabilistic matching} 
	Let $N_1$ and $N_2$ be disjoint sets of agents. 
	Each agent $i\in N_1$ has a preference relation $\hat\pref_i$ over agents in $N_2$ and likewise for agents in $N_2$.
	A deterministic matching is a subset $\mu$ of $N_1\times N_2$ such that $(i,j),(i,j')\in\mu$ implies $j = j'$ and $(i,j),(i',j)\in\mu$ implies $i = i'$.
	Let $A$ be the set of all deterministic matchings.
	The preference relation $\pref_i$ of $i\in N_1$ over $A$ has $\mu\pref_i \mu'$ if and only if $(i,j)\in\mu$ and $(i,j')\in\mu'$ with $j\mathrel{\hat\pref_i} j'$ or $(i,j)\in\mu'$ for no $j\in N_2$; preferences for agents in $N_2$ are defined analogously.
	A random matching is an element of the unit simplex $\Delta(A)$ and $X$ specifies a set of feasible random matchings. 
	We can apply VER to the problem with two-sided preferences and let both sides eat simultaneously.
This approach is promising because many standard mechanisms for two-sided matching are asymmetric in that they treat the two sides differently.  
If we do not impose any stability constraints and both sides have dichotomous preferences, then VER is equivalent to \citeauthor{BoMo04a}'s~\citeyearpar{BoMo04a} egalitarian rule.

We assumed that agents have complete preference orders over alternatives. 
VER and its properties extend to preferences given by partial orders. 
The change that is required is the same that \citet{KaSe06a} suggested for adapting the extended probabilistic serial rule to partial orders. 
Instead of trying to increase the probability for most preferred alternatives, agents try to increase the probability for those alternatives that are not strictly dominated by any other alternative.

 	\section*{Acknowledgements}
 	The authors thank Fuhito Kojima, Debasis Mishra, Herv{\'{e}} Moulin, Barton Lee, and Arunava Sen for helpful comments. They also thank the participants of the following events where the paper was presented:  Indian Statistical Institute Seminar Series; COMSOC Video Seminar; INFORMS Workshop on Market Design 2021; and the 1st IJCAI-PRICAI Workshop on Applied Mechanism Design.
 	Florian Brandl acknowledges support by the Deutsche Forschungsgemeinschaft under grant BR 5969/1-1.

\appendix
\section*{APPENDIX}

\section{Proofs From \Cref{sec:ver}}

\verterminates*
\begin{proof}
	First, we show that the optimization problem in Line~\ref{ln:delta} always has a solution.
	Second, we show that from each round to the next, either some agent is removed from the set of active agents $N'$ or some agent moves to a less preferred equivalence class.
	
	Let $k$ be the index of a round in the algorithm.
	For the first statement, we have to show that there exists $(p,\delta)$ that satisfies the constraints in Line~\ref{ln:delta} in round $k$.
	If $k = 0$, this is trivial since $\pi(i,E) = 0$ for all $i\in N$ and $E\in\mathcal E_i$ and $X$ is non-empty.
	If $k > 0$, let $(p, \delta^*)$ be an optimal solution to the optimization problem computed in round $k-1$.
	It follows from the definition of $\pi^k$ (at the end of round $k-1$) that $(p,\delta^*)$ satisfies all constraints of the optimization problem in round $k$. 
	Thus, the set of feasible points in round $k$ is non-empty.
	Since $X$ is closed, the problem has an optimal solution.
	
	For the second statement, denote by $\mathcal F_i^k$ the set of equivalence classes available to agent $i$ in round $k$ (cf. Line~\ref{ln:oi}), by $N'_k$ the set of agents so that $\mathcal F_i^k\neq\emptyset$, and by $E_i^k = \max_{\pref_i}\mathcal F_i^k$ for $i\in N'$ their most preferred equivalence classes.
	Define $\mathcal F_i^{k+1}$, $N'_{k+1}$, and, for $i\in N'_{k+1}$, $E_i^{k+1}$ similarly.
	Note that $N'_k\neq\emptyset$ since we would not have reached round $k+1$ otherwise.
	Moreover, $N'_k\neq\emptyset$ implies that $\delta^*$ as computed in Line~\ref{ln:delta} in round $k$ is strictly positive by definition of $\mathcal F_i^k$.

	Now if $N'_{k+1} = \emptyset$, some agent is removed from the set of active agents and there is nothing left to show.
	So assume $N'_{k+1}\neq\emptyset$.
	We want to show that for $i\in N'$, $\mathcal F_i^{k+1}\subset\mathcal F_i^k$.
	If $E\in\mathcal F_i^{k+1}$, there is $p\in X$ such that $p(j,E')\ge \pi^{k+1}(j,E')$ for all $j\in N$ and $E'\in\mathcal E_j$, $p(i,E) > \pi^{k+1}(i,E)$, and $p(j,E_j) > \pi^{k+1}(j,E_j)$ for all $j\in N'_{k+1}$ with $j < i$.
	These properties of $p$, the definition of $\pi^{k+1}$, and the fact that $\delta^* > 0$ in round $k$ imply that $E\in\mathcal F_i^k$.
	So we get $E_i^k\succsim_i E_i^{k+1}$ for all $i\in N'_{k+1}$.
	 If $N'_{k+1} = N'_k$, it follows from the choice of $\delta^*$ that this preference is strict for at least one $i\in N'_{k+1}$.
	 Otherwise $N'_{k+1}\subsetneq N'_k$.
	 It follows that the algorithm terminates after at most $m\cdot n$ rounds.
	 Clearly, the returned allocation $p$ is in $X$.

\end{proof}

\constrainedefficient*
\begin{proof}
	Let $\apref$ be a preference profile and $p$ be the allocation returned by VER for $X$ and $\apref$.
	For every allocation $q$ and $t\ge 0$, let $q^t$ be the allocation where every agent $i$ receives a prefix (according to her preferences) of $q(i)$ summing to $\min\{t, q(i,O)\}$.
	Formally, $q^t(i,O) = \min\{t, q(i,O)\}$ and, for $E_i\in\mathcal E_i$, $q^t(i,E_i) > 0$ implies $q^t(i,E_i') = q(i,E_i')$ for all $E_i'\succ_i E_i$.
	
	Assume there is an allocation $q\in X$ that lexicographically dominates $p$.
	Let
	\[
	t^* = \max\{t\ge 0\colon p^t(i) \sim_i^\dl q^t(i) \text{ for all $i\in N$}\}
	\]
	Let $i^*$ be the lexicographically first agent so that $q^{t+\epsilon}(i) \succ_i^\dl p^{t+\epsilon}(i)$ for all $\epsilon > 0$.
	
	Observe that $p^{t^*}(i^*,E_i) = p(i^*,E_i)$ whenever $p^{t^*}(i^*,E_i) > 0$.
	To see this, let $E_i^*$ be $i^*$'s least preferred equivalence class those with $p^{t^*}(i^*,E_i) > 0$.
	It certainly holds that $p^{t^*}(i^*,E_i) = p(i^*,E_i)$ for all $E_i$ with $E_i\succ_{i^*} E_i^*$ by definition of $p^{t^*}$.
	Now if $p(i^*,E_i^*) > p^{t^*}(i^*,E_i^*)$, let $0 < \epsilon < p(i^*,E_i^*) - p^{t^*}(i^*,E_i^*)$.
	We have $p^{t^*+\epsilon}(i^*,E_i) = q^{t^*+\epsilon}(i^*,E_i)$ for all $E_i\succ_{i^*} E_i^*$ by the choice of $t^*$.
	Moreover, by the choice of $\epsilon$, $p^{t^*+\epsilon}(i^*,O) = q^{t^*+\epsilon}(i^*,O)$ and so $p^{t^*+\epsilon}(i^*,E_i^*) = q^{t^*+\epsilon}(i^*,E_i^*)$.
	But this contradicts $q^{t^*+\epsilon}(i^*) \succ_{i^*}^\dl p^{t^*+\epsilon}(i^*)$.
	It follows that agent $i^*$ moves to a less preferred equivalence class after securing $t^*$ of equivalence classes at least as good as $E_i^*$.
	That is, there is a round $k$ in \Cref{algo:closedver} such that $p^{t^*}(i,E_i) = \pi(i,E_i)$ for all $i\in N$ and $E_i\in\mathcal E_i$.
	
	Let $N'$ be the set of active agents at the end of round $k$, that is, all agents who increase their guarantee for some equivalence class in round $k$.
	For every $i\in N'$, let $E_i$ be as determined by Line~\ref{ln:oi} in round $k$.
	Since $p$ is the outcome of VER, it follows that $p(i,E_i) > \pi^k(i, E_i)$ for all $i\in N'$
	Moreover, $q(i,E_i) > \pi^k(i, E_i)$ for all $i < i^*$ (by the choice of $i^*$) and, since $q^{t^* + \epsilon} \succ_{i^*}^\dl p^{t^*+\epsilon}$, there is $E_{i^*}'\succ_i E_{i^*}$ such that $q(i^*,E_{i^*}') > p(i^*,E_{i^*}') \ge \pi^k(i^*,E_{i^*}')$.
	This contradicts the choice of $E_{i^*}$ since $q$ is a witness that $i^*$ could have chosen $E_{i^*}'$ instead of $E_{i^*}$.
\end{proof}

\polynomial*
\begin{proof}
	Let $X = \bigcup_{r = 1}^s P_r$ so that $P_r$ is a polytope for each $r$.
	The only computationally non-trivial steps are determining the set $\mathcal F_i$ in Line~\ref{ln:oi} and the computation of $\delta^*$ in Line~\ref{ln:delta}.
	
	First, observe that given a round $k$ of the algorithm, an agent $i$, an equivalence class $E\in\mathcal E_i$, and a polytope $P_r$, the problem
	\begin{align*}
		\epsilon^* &= \max_{\epsilon,p} \epsilon \text{ s.t. }&\\
		p(j,E_j)&\ge \pi^k(j,E_j) + \epsilon & \text{$\forall$ $j\in N'$ with } j<i\\
		p(i,E)&\ge \pi^k(i,E) + \epsilon\\
		p(j,E')&\ge \pi^k(j,E')&\forall j\in N\text{ and } E'\in \mathcal E_j\\
		p&\in P_r
	\end{align*}
	can be solved in polynomial time in the number of linear inequalities used to describe $P_r$.
	If $\epsilon^*>0$, agent $i$ can increase her probability for $E$.
	By solving this problem for every $r$ and every equivalence class $E\in\mathcal E_i$, we can determine the most-preferred equivalence class $E_i$ of which agent $i$ can increase her share in polynomial time in the number of inequalities used to describe $X$.
	
	Second, the linear program
	\begin{align*}
	\delta_r^*&=\max_{\delta,p} \delta  \text{ s.t. }& \\
	p(i,E_i)&\ge  \pi^k(i,E_i) + \delta&\forall i\in N'\\
	p(i,E)&\ge \pi^k(i,E) &\forall i\in N\text{ and } E\in \mathcal E_i\\
	p&\in P_r
	\end{align*}
	can be solved in polynomial time for every $P_r$. 
	Since $\delta^* = \max_r\delta^*_r$, the value of $\delta^*$ in Line~\ref{ln:delta} can be computed in polynomial time.
	
	\Cref{thm:verterminates} ensures that the number of iterations of the while-loop in Line~\ref{ln:while} is bounded by $m\cdot n$.
\end{proof}

\ordinalegal*
	\begin{proof}		
Consider an allocation $q$ that is OE and an allocation $p$ that is returned by VER. Consider $t^\uparrow(q) = (t_i^j(q))_{i,j}$ the signature vector of $q$ and $t^\uparrow(p) = (t_i^j(p))_{i,j}$ the signature vector of $p$.

Consider that during the run of VER, lower bound constraints of the following form are added: $p(E_i^j)\geq \lambda$. The set of such constraints can alternatively be written as lower bounds on the upper contour set as follows: $p(\bigcup_{\ell=1}^j E_i^{\ell})\geq \lambda'$. Equivalently, they can be written as $t_i^j\geq \lambda'$. 
At the start of round $k$, we denote by $g(i,k)$ the number of the first equivalence class $E_i^{g(i,k)}$ of agent $i$ for which the lower bound has not been fixed. 

We prove by induction on the rounds $k$ of the algorithm that VER finds the largest $s$ such that $t_i^{g(i,k)}\geq s$ and that the following values are present in $t^\uparrow(q)$: $\min_{i'\in N} p(\bigcup_{\ell=1}^{g(i',k)}E_{i'}^{\ell}))$ as well as $p(\bigcup_{\ell=1}^{k'}E_{i'}^{\ell}))$ for $i\in N$ and $k'<g(i',k)$.

For $k=0$, in the first round, VER tries to maximizes the lower bound on $E_i^{g(i,1)}$ for all $i\in N$. 
Suppose that $p^1(E_i^{g(i,1)})\geq \lambda'$ for all $i\in N$. Then $\lambda'$ is by definition $\delta^*$ as computed by VER. It also follows that for all $i\in N$, $p^1(E_i^{j})=0$ for all $j<g(i,1)$. Hence, we have established that the minimum non-zero entry in the vector $t^\uparrow(p)$ is the same as the minimum non-zero entry in vector $t^\uparrow(q)$. By convexity of the feasible region it follows that it is not possible to have an allocation in which $p(\bigcup_{\ell=1}^j E_i^{\ell})>0$ for any $j<g(i,1)$.
Therefore, the corresponding entries $t_i^{\ell}=0$ are also entries in $t^\uparrow(q)$.  

Now suppose $k$ rounds have passed. By the induction hypothesis, for each $k'\leq k$, $t_i^{g(i,k')}$ is present in the vector $t^\uparrow(q)$. We also note that VER has fixed a weight for each of $\bigcup_{\ell=1}^j E_i^{\ell}$ where $j<g(i,k+1)$.
At this point, VER computes the largest $\delta$ that can be additionally guaranteed for each $E_i^{g(i,k+1)}$. Equivalently, it computes the largest $s$ such that
$p(\bigcup_{\ell=1}^{g(i,k+1)} E_i^{\ell})\geq s$ for all $i\in N$. Then $\min_{i'\in N} p(\bigcup_{\ell=1}^{g(i',k+1)}E_{i'}^{\ell}))$ is the optimised weight of the next heavy upper contour set. Hence $\min_{i'\in N} p(\bigcup_{\ell=1}^{g(i',k+1)}E_{i'}^{\ell}))$ is 
present in $t^\uparrow(q)$. 
\end{proof}

\symmetric*
\begin{proof}
	Let $p$ be the outcome of VER for the set $X$ and the profile $\apref$.
	We prove by induction that for each round $k$ of \Cref{algo:closedver}, $\pi^k(i,E) = \pi^k(i',E)$ for all $E\in\mathcal E_i = \mathcal E_i'$. 
	It then follows that $p(i,E) = p(i',E)$ for all $E\in\mathcal E_i$.
	
	For $k = 0$, this is clear since $\pi^0(j,E) = 0$ for all $j\in N$ and $E\in \mathcal E_i$.
	Now let $k$ be arbitrary and assume that $\pi^k(i,E) = \pi^k(i',E)$ for all $E\in\mathcal E_i$.
	Let $N'$ be the set active agents at the beginning of round $k$.
	We show that $\mathcal F_i = \mathcal F_{i'}$, where the $\mathcal F_j$ are determined as in Line~\ref{ln:oi}.
	 
	Assume that $i < i'$.
	Since $X$ is symmetric for $i$ and $i'$, it follows that $\mathcal F_{i'}\subseteq \mathcal F_i$.
	If $\mathcal F_i = \emptyset$, we are done.
	Otherwise, let $j = \max\{j'\in N'\colon j' < i' \text{ and } \mathcal F_{j'}\neq\emptyset\}$ (note that $j\ge i$).
	By definition of $E_j$, there is an allocation $p\in X$ such that $p(j',E) \ge \pi^k(j',E)$ for all $j'\in N$ and $E\in\mathcal E_{j'}$ and $p(j',o_j') > \pi^k(j',E_{j'})$ for all $j'\in N'$ with $j' \le j$ and $\mathcal F_{j'}\neq\emptyset$.
	Since $X$ is symmetric and convex for $i$ and $i'$, the allocation $p' = \frac12 p + \frac12 p^{(ii')}$ is in $X$.
	Thus defined, $p'$ satisfies the constraints on $p$ above. 
	For $j'\neq i,i'$, this is obvious since $p'(j') = p(j')$.
	For $i$ and $i'$, it follows from $p(i,E_i) > \pi^k(i,E_i)$ and $\pi^k(i,E) = \pi^k(i',E)$.
	Clearly, $p'(i) = p'(i')$ and so $p'(i',E_i) > \pi^k(i',E_i)$.
	It follows that $E_i\in\mathcal F_{i'}$.
	Thus, $E_i = E_{i'}$ and so $\pi^{k+1}(i,E) = \pi^{k+1}(i',E)$ for all $E\in\mathcal E_i$.
\end{proof}

\dlthm*
  \begin{proof}
	  Our arguments in the proof of \Cref{thm:constrainedefficient} go through to show that the outcome of any instance of one-at-a-time-VER is constrained \dl-efficient. 
	  
	  For the converse, assume that $p$ is constrained \dl-efficient with respect to $X$.
	  Our arguments are similar to those used in the proof of Theorem 8 by \citet{AzSt14a}.
	  We show that if $\pi$ is an allocation to equivalence classes so that $p(i,E_i) \ge \pi(i,E_i)$ for all $i\in N$ and $E_i\in\mathcal E_i$ with at least one strict inequality, then there is an agent $i^*$ so that the most-preferred equivalence class $E_{i^*}$ for which $i^*$ can increase her probability has $p(i^*,E_{i^*}) > \pi(i^*,E_{i^*})$.
	  Then, by allowing $i^*$ to eat, we can get closer to our target allocation $p$ and the claim follows. 
	  
	  Suppose we have $\pi$ as above.
	  For every $i\in N$, let $\mathcal F_i = \{E\in \mathcal E_i\colon \text{there is $q\in X$ so that } q(j,E_j)\ge \pi(j,E_j)\text{ for all } j\in N\text{ and } E_j\in \mathcal E_j\text{ and } q(i,E) > \pi(i,E)\}$. 
	  Let $N' = \{i\in N\colon \mathcal F_i\neq\emptyset\}$ and for $i\in N'$, let $E_i = \max_{\pref_i}\mathcal F_i$.
	  Note that $N'$ is non-empty since we assume that $p(i,E) > \pi(i,E)$ for some $i\in N$ and $E\in\mathcal E_i$.
	  Since $X$ is convex, we can find $p^*\in X$ so that $p^*(i,E)\ge \pi(i,E)$ for all $i\in N$ and $E\in \mathcal E_i$ and $p^*(i,E_i) > \pi(i,E_i)$ for all $i\in N'$.
	  
	  Now if $p(i,E_i') > \pi(i,E_i')$ for $E_i'\in\mathcal E_i$, then $E_i\pref_i E_i'$ since $p$ is a witness that $E_i'\in\mathcal F_i$.
	  Since $p$ is \dl-efficient, it is not \dl-dominated by $p^*$.
	  So either $p(i) \sim_i^\dl p^*(i)$ for all $i\in N$ or there is $i\in N$ such that $p(i)\succ_i^\dl p^*(i)$.
	  In the first case we can choose $i^*\in N'$ arbitrarily.
	  In the second case, choose $i\in N$ so that $p(i)\succ_i^\dl p^*(i)$.
	  So we can find $E_i'$ with $p(i,E_i') > p^*(i,E_i')$ and $p(i,E_i'') = p^*(i,E_i'')$ for all $E_i''\succ_i E_i'$.
	  It follows that $E_i'\in \mathcal F_i$ and so $i\in N'$ and $E_i\pref_i E_i'$.
	  The latter implies that $p(i,E_i) \ge p^*(i,E_i) > \pi(i,E_i)$.
	  We can thus choose $i^* = i$.
	    \end{proof}

\section{Proofs From \Cref{sec:stable}}\label{sec:proofssection5}

\stableandefficient*
\begin{proof}
	Stability and constrained \sd-efficiency of the VER allocation follow from Theorems~\ref{thm:verterminates} and~\ref{thm:constrainedefficient}.
	
	To show that VER satisfies limited equal treatment of equals, let $i,i'\in N$ such that ${\pref_i} = {\pref_i'}$, $\capa(i) = \capa(i')$, and $i\sim_o i'$ for all $o\in O$.
	We prove that the set of S-stable allocations is symmetric and convex for $\{i,i'\}$.
	Then \Cref{thm:symmetric} yields the desired conclusion.
	Symmetry is obvious for all stability notions.
	Also, for ex-post stability, fractional stability, and claimwise stability, the set of stable allocations is convex and thus convex for $\{i,i'\}$.
	The set of ex-ante stable allocations is not in general convex.
	It is convex for $\{i,i'\}$, however, as we show now.

	Let $p$ be ex-ante stable and $\lambda\in[0,1]$; let $q = \lambda p + (1-\lambda) p^{(ii')}$.
	Note that $q(j) = p(j)$ for all $j\neq i,i'$.
	Thus, it suffices to check envy-freeness for pairs of agents in $\{i,i'\}\times N\setminus\{i,i'\}$.
	Let $j\in N\setminus\{i,i'\}$.
	First, if $i\succ_o j$ and $q(j,o) > 0$, then $p(j,o) > 0$ since $p(j) = p^{(ii')}(j)$.
	Thus, $\sum_{o'\succsim_i o} p(i,o') = \sum_{o'\succsim_{i'} o} p(i',o') = \capa(i) = \capa(i')$.
	It follows that $\sum_{o'\succsim_{i} o} q(i,o') = \capa(i)$ and similarly for $i'$.
	Second, if $j\succ_o i$ and $q(i,o) = q(i',o) > 0$, then without loss of generality, $p(i,o) > 0$. 
	Hence, $\sum_{o'\succsim_j o} p(j,o') = \sum_{o'\succsim_j o} p^{(ii')}(j,o') = \capa(j)$.
	We get $\sum_{o'\succsim_j o} q(j,o') = \capa(j)$ as desired.
	Together with the fact that $q(j) = p(j)$ for $j\neq i,i'$, this shows that $q$ is ex-ante stable.
\end{proof}

\agentoptimal*
\begin{proof}
	First notice that if preferences and priorities are strict, ex-post stability and fractional stability coincide.
	We use that the sets of ex-ante stable allocations and of ex-post stable allocations are lattices when join and meet are defined via stochastic dominance.
	This has been shown by \citet[][Theorem 5]{AlGa03a} for ex-ante stability (and a more general class of preferences) and by \citet{JNO20a} for ex-post stability. 
	The upper bound (with respect to the join operation) of both lattices is the agent-optimal deterministic stable allocation, which thus stochastically dominates every other ex-post or ex-ante stable allocation according the the agents' preferences.
	Since by \Cref{cor:stableandefficient} the outcome of $S$-VER is \sd-efficient among $S$-stable allocations, it follows that $S$-VER returns the agent-optimal deterministic stable allocation when $S$ is ex-ante, ex-post, or fractional stability.

	Now consider VER for claimwise stability.
	Let $p$ be the outcome of CWS-VER for the profile $\pss$.
	We show that $p$ is a deterministic stable allocation.
	It then follows from the fact that $p$ is \sd-efficient among (claimwise) stable allocations that $p$ is the agent-optimal deterministic stable allocation.
	Let $i,j\in N$.
	In the first round of \Cref{algo:closedver} ($k = 0$), agents $i$ and $j$ eat objects $o_i$ and $o_j$, respectively. (Since preferences are strict, agents eat objects instead of equivalence classes.)
	We have shown in \Cref{thm:verterminates} that agents eat objects in decreasing order of their preference, so $o_i\succsim_i o$ for all $o\in O$ with $p(i,o) > 0$ and similarly for $j$.
	Suppose $j\succ_{o_i} i$.
	If $o_i\succsim_j o_j$, then $0 = \sum_{o\succ_j o_i} p(j,o) < p(i,o_i)$, which contradicts claimwise stability of $p$.
	Thus, either $o_j\succ_j o_i$ or $i\succ_{o_i} j$. 
	Since this holds for all $i,j$, the deterministic allocation $q$ that assigns $o_i$ to $i$ for all $i$ is claimwise stable. 
	Moreover, $q$ weakly stochastically dominates $p$.
	Since, by \Cref{prop:efficientandstable}, $p$ is $\sd$-efficient among claimwise stable allocations, it follows that $p = q$.
	For deterministic allocations and strict priorities, claimwise stability reduces to stability and so $p$ is the agent optimal deterministic stable allocation.
	
	Alternatively, one can show that CWS-VER is equivalent to the claimwise probabilistic serial rule defined by \citet{Afac18a} and apply his Proposition~4.
\end{proof}

\stablebothsides*

\begin{proof}
	First we consider fractional stability.
	We show that on the set of fractionally stable allocations, (one-sided) constrained \sd-efficiency implies two-sided \sd-efficiency.
	Then the claim follows from \Cref{cor:stableandefficient}. 
	Let $p$ an allocation that is constrained \sd-efficient among fractionally stable allocations.
	Assume that $q\neq p$ two-sided stochastically dominates $p$.
	Frist, observe that $q$ is fractionally stable since the inequalities defining fractional stability are preserved under improvements with respect to stochastic dominance.
	By assumption, all agents weakly prefer $q$ to $p$ according to stochastic dominance.
	Since $q\neq p$ and preferences are strict, this preference is strict for at least one agent.
	In summary, $q$ is a fractionally stable allocation that stochastically dominates $p$, which is a contradiction.
	
	Ex-ante stability requires more work.
	Let $\ps$ be a profile and $p$ be an allocation that is ex-ante stable and \sd-efficient among ex-ante stable allocations.
	By \Cref{cor:stableandefficient}, it suffices to show that $p$ is two-sided \sd-efficient.

	If $p$ is not two-sided \sd-efficient, there exists a cycle $(i_0,o_0),(i_1,o_1),\dots,(i_k,o_k)$ such that $(i_0,o_0) = (i_k,o_k)$ and for all $l = 1,\dots, k$, $o_l\succ_{i_l} o_{l-1}$, $i_{l-1}\succsim_{o_{l-1}} i_l$, and $p(i_l,o_{l-1})> 0$ \citep[cf.][Proposition 1]{DoYi16a}.
	For $\epsilon\in(0,\min\{p(i,o)\colon i\in N\text{ and } o\in O\})$, let $q$ be equal to $p$ except that $q(i_l,o_l) = p(i_l,o_l) + \epsilon$ and $q(i_l,o_{l-1}) = p(i_l,o_{l-1}) - \epsilon$ for all $l = 1,\dots,k$.
	Thus, agent $i_l$ passes a fraction of $o_{l-1}$ on to agent $i_{l-1}$. 
	Note that $q$ stochastically dominates $p$ for both sides and at least one agent $i\in N$ strictly prefers $q$ to $p$ since preferences are strict.
	
	We show that $q$ is ex-ante stable.
	Let $i,j\in N$ and $o\in O$ with $i\succ_o j$ and $q(j,o) > 0$.
	If $p(j,o) > 0$, ex-ante stability of $p$ implies $p(i,o') = 0$ for all $o'$ with $o\succ_i o'$.
	Since $q(i)\succsim_i^{\sd} p(i)$, it follows that $q(i,o') = 0$ for all $o'$ with $o\succ_i o'$.
	Otherwise, $j = i_{l-1}$ and $o = o_{l-1}$ for some $l = 1,\dots, k$.
	Now $p(i_{l},o_{l-1}) \ge\epsilon > 0$ and $i_{l-1}\succsim_{o_{l-1}} i_{l}$.
	It follows that $i\succ_{o_{l-1}} i_l$.
	Since $p$ is ex-ante stable, $p(i,o') = 0$ for all $o'$ with $o_{l-1}\succ o'$.
	And again, since $q(i)\succsim_i^{\sd} p(i)$, it follows that $q(i,o') = 0$ for all $o'$ with $o_{l-1}\succ_i o'$.
	So $q$ is ex-ante stable.
	Thus, $p$ is not \sd-efficient among ex-ante stable allocations, which contradicts the assumption.
\end{proof}

\sdclaimwisedisjoint*

			\begin{proof}
				Consider the following instance. 

				\[
					\begin{array}{cccc}
						\succ_1\colon & b&a&c\\
						\succ_2\colon & a&b&c\\			
						\succ_3\colon & a&b&c\\
					\end{array}
					\qquad\qquad
					\begin{array}{cccc}
						\succ_a\colon & 1&3&2\\
						\succ_b\colon & 2&1&3\\			
						\succ_c\colon & 2&1&3\\
					\end{array}
				\]

				The only deterministic stable allocation is the following one. 

			\begin{center}
			 \begin{blockarray}{ccccccccccc}

			 		&&\matindex{$a$}&\matindex{$b$}& \matindex{$c$}&\\
			 	    \begin{block}{c(cccccccccc)}
			 			\matindex{$1$}& &$1$&$0$&$0$& \\
			 			\matindex{$2$}& &$0$&$1$&$0$& \\
						\matindex{$3$}& &$0$&$0$&$1$& \\
			 	    \end{block}
			 	  \end{blockarray}
				  \end{center}
  
				  We  show that it is also the only claimwise stable allocation. Consider any claimwise stable allocation $p$. We first claim that $p(2,a)=0$. If $p(2,a)>0$, agent 1 will have a justified claim against agent 2 for object $a$. 
  
				  \begin{center}
			      $p=$\begin{blockarray}{ccccccccccc}
			       	 
			       		&&\matindex{$a$}&\matindex{$b$}& \matindex{$c$}&\\
			       	    \begin{block}{c(cccccccccc)}
			       			\matindex{$1$}& &?&?&?& \\
			       			\matindex{$2$}& &$0$&?&?& \\
			      			\matindex{$3$}& &?&?&?& \\
			       	    \end{block}
			       	  \end{blockarray}
				    \end{center}
  
				Next, we claim that $p(1,b)=0$. Since $p(2,a)=0$, agent $2$ cannot let agent $1$ get any part of $b$ or else $2$ will have a justified claim against agent $1$ for object $b$.
				 \begin{center}
			   $p=$\begin{blockarray}{ccccccccccc}
			    	 
			    		&&\matindex{$a$}&\matindex{$b$}& \matindex{$c$}&\\
			    	    \begin{block}{c(cccccccccc)}
			    			\matindex{$1$}& &?&$0$&?& \\
			    			\matindex{$2$}& &$0$&?&?& \\
			   			\matindex{$3$}& &?&?&?& \\
			    	    \end{block}
			    	  \end{blockarray}
				  \end{center}

				We now claim that $p(3,a)=0$. If $p(3,a)>0$, then we know that $p(1,b)=0$ so 1 cannot let $3$ take any part of $a$ or else it will have a justified claim against $3$ for object $a$.

				 \begin{center}
			    $p=$\begin{blockarray}{ccccccccccc}
			     	 
			     		&&\matindex{$a$}&\matindex{$b$}& \matindex{$c$}&\\
			     	    \begin{block}{c(cccccccccc)}
			     			\matindex{$1$}& &?&$0$&?& \\
			     			\matindex{$2$}& &$0$&?&?& \\
			    			\matindex{$3$}& &$0$&?&?& \\
			     	    \end{block}
			     	  \end{blockarray}
					    \end{center}
		
					    Since the matrix is bistochastic, we complete some columns and rows.

					 \begin{center}
				      $p=$\begin{blockarray}{ccccccccccc}
			  	       	 
			  	       		&&\matindex{$a$}&\matindex{$b$}& \matindex{$c$}&\\
			  	       	    \begin{block}{c(cccccccccc)}
			  	       			\matindex{$1$}& &$1$&$0$&$0$& \\
			  	       			\matindex{$2$}& &$0$&?&?& \\
			  	      			\matindex{$3$}& &$0$&?&?& \\
			  	       	    \end{block}
			  	       	  \end{blockarray}
						    \end{center}

				Next, we argue that $p(3,b)=0$. If 	$p(3,b)>0$, then we know that $p(2,b)=0$ so agent $2$ will not let $3$ get any part of $b$ or else $2$ will have a justified claim against $3$ for object $b$. Hence, 	

					 \begin{center}
				      $p=$\begin{blockarray}{ccccccccccc}
			  	       	 
			  	       		&&\matindex{$a$}&\matindex{$b$}& \matindex{$c$}&\\
			  	       	    \begin{block}{c(cccccccccc)}
			  	       			\matindex{$1$}& &$1$&$0$&$0$& \\
			  	       			\matindex{$2$}& &$0$&?&?& \\
			  	      			\matindex{$3$}& &$0$&$0$&?& \\
			  	       	    \end{block}
			  	       	  \end{blockarray}
						     \end{center}

				We can now complete the matrix.

				 	 \begin{center}
			  $p=$\begin{blockarray}{ccccccccccc}
			   	 
			   		&&\matindex{$a$}&\matindex{$b$}& \matindex{$c$}&\\
			   	    \begin{block}{c(cccccccccc)}
			   			\matindex{$1$}& &$1$&$0$&$0$& \\
			   			\matindex{$2$}& &$0$&$1$&$1$& \\
			  			\matindex{$3$}& &$0$&$0$&$1$& \\
			   	    \end{block}
			   	  \end{blockarray}
				     \end{center}
  
				  We have established that $p$ is the only claimwise stable allocation. However, it is not \sd-efficient.

				  In particular, the allocation 
  
				  	 	 \begin{center}
			   	\begin{blockarray}{ccccccccccc}
			 	 
			 		&&\matindex{$a$}&\matindex{$b$}& \matindex{$c$}&\\
			 	    \begin{block}{c(cccccccccc)}
			 			\matindex{$1$}& &$0$&$1$&$0$& \\
			 			\matindex{$2$}& &$1$&$0$&$0$& \\
						\matindex{$3$}& &$0$&$0$&$1$&\\
			 	    \end{block}
			 	  \end{blockarray}
				     \end{center}
				  	  \sd-dominates $p$.

				\end{proof}

\strategyproof*
\begin{proof}
	In all profiles in the proof, the sets of S-stable allocations with be the same for $S\in\{\text{ex-ante, ex-post, FS}\}$.
	Thus, it will prove the statement for all three stability notions at once.
	
	Let $f$ be a mechanism that is weakly \sd-strategyproof, S-stable, and S-constrained efficient.
	Consider the following profile $\ps$.
	\[
		\begin{array}{cccc}
			\succ_1\colon & b&c&a\\
			\succ_2\colon & b&c&a\\			
			\succ_3\colon & a&b&c\\
		\end{array}
		\qquad\qquad
		\begin{array}{cccc}
			\pref_a\colon & \{1,2\}&3\\
			\pref_b\colon & 3&\{1,2\}\\
			\pref_c\colon & 3&\{1,2\}\\
		\end{array}
	\]					
	There are two S-stable and S-constrained efficient deterministic allocations.
	\begin{center}
		\begin{blockarray}{ccccccccccc}
			 &&\matindex{$a$}&\matindex{$b$}& \matindex{$c$}&\\
		  \begin{block}{c(cccccccccc)}
				\matindex{$1$}& &$0$&$1$&$0$& \\
				\matindex{$2$}& &$0$&$0$&$1$& \\
				\matindex{$3$}& &$1$&$0$&$0$& \\
		  \end{block}
		\end{blockarray}
		\qquad\qquad
		\begin{blockarray}{ccccccccccc}
			 &&\matindex{$a$}&\matindex{$b$}& \matindex{$c$}&\\
		  \begin{block}{c(cccccccccc)}
				\matindex{$1$}& &$0$&$0$&$1$& \\
				\matindex{$2$}& &$0$&$1$&$0$& \\
				\matindex{$3$}& &$1$&$0$&$0$& \\
		  \end{block}
		\end{blockarray}
	\end{center}
	Thus, all S-stable and S-constrained efficient allocations are of the following form. 
	(Here we use that the set of S-stable allocations in the same for all three stability notions.)
	 
	\begin{center}
		$p^{\alpha}=$
		\begin{blockarray}{ccccccccccc}
			 &&\matindex{$a$}&\matindex{$b$}& \matindex{$c$}&\\
		  \begin{block}{c(cccccccccc)}
				\matindex{$1$}& &$0$&$1-\alpha$&$\alpha$& \\
				\matindex{$2$}& &$0$&$\alpha$&$1-\alpha$& \\
				\matindex{$3$}& &$1$&$0$&$0$& \\
		  \end{block}
		\end{blockarray}
	\end{center}
	for some $\alpha\in[0,1]$.
	Suppose $f(\ps) = p^\alpha$ for some $\alpha > 0$.
	Then, agent 1 can misreport by swapping $c$ and $a$ resulting in the preferences $\aprefs'$ below. 
	\[
		\begin{array}{cccc}
			\succ_1'\colon & b&a&c\\
			\succ_2'\colon & b&c&a\\			
			\succ_3'\colon & a&b&c\\
		\end{array}
	\]
	S-stability and S-constrained efficiency imply that $f((\aprefs',{\succsim_O})) = p^0$.
	
	Note that $p^0(1)\succ_i^{\sd} p^\alpha$ and so agent 1 can successfully manipulate in the profile $\ps$, which contradicts strategyproofness of $f$.
	If $1-\alpha > 0$, we can use a symmetric argument where agent $2$ misreports. 				  
\end{proof}

\begin{example}\label{ex:fractionalnotex-post}\normalfont
Fractional stability is strictly weaker than ex-post stability for weak priorities and strict preferences.
To see this, consider the following example.
Let $N=\{1,2,3,4,5\}$, $O=\{a,b,c,d,e\}$, and $\capa(i) = 1$ for all $i\in N$. 
The preferences and priorities are as follows.
\begin{center}
\begin{tabular}{lccccc}
$\succ_1$:&$c$&$d$&$e$&$a$&$b$\\
$\succ_2$:&$c$&$d$&$e$&$a$&$b$\\
$\succ_3$:&$a$&$b$&$d$&$c$&$e$\\
$\succ_4$:&$a$&$b$&$c$&$d$&$e$\\
$\succ_5$:&$a$&$b$&$c$&$d$&$e$\\
\end{tabular}
\quad\quad
\begin{tabular}{lccc}
$\succsim_a$:&$[1,2,3,4,5]$\\
$\succsim_b$:&$[1,2,3,4,5]$\\
$\succsim_c$:&$[3,4,5]$&$[1,2]$\\
$\succsim_d$:&$[3,4,5]$&$[1,2]$\\
$\succsim_e$:&$[1,2,3,4,5]$\\
\end{tabular}
\end{center}
Then the allocation
\begin{center}
	$p=$
				 \begin{blockarray}{ccccccccccc}

				 		&&\matindex{$a$}&\matindex{$b$}& \matindex{$c$}& \matindex{$d$}& \matindex{$e$}\\
				 	    \begin{block}{c(cccccccccc)}
				 			\matindex{$1$}& &$0$&$0$&$\nicefrac{1}{3}$&$\nicefrac{1}{3}$&$\nicefrac{1}{3}$& \\
				 			\matindex{$2$}& &$0$&$0$&$\nicefrac{1}{3}$&$\nicefrac{1}{3}$&$\nicefrac{1}{3}$& \\
							\matindex{$3$}& &$$\nicefrac{1}{3}$$&$$\nicefrac{1}{3}$$&$0$&$\nicefrac{1}{3}$ &$0$&\\
							\matindex{$4$}& &$\nicefrac{1}{3}$&$\nicefrac{1}{3}$&$\nicefrac16$& $0$&$\nicefrac16$&\\
							\matindex{$5$}& &$\nicefrac{1}{3}$&$\nicefrac{1}{3}$&$\nicefrac16$&$0$ &$\nicefrac16$&\\
				 	    \end{block}
				 	  \end{blockarray}
					  \end{center}
is fractionally stable, but not ex-post stable.

To see that $p$ is fractionally stable, observe that \eqref{eq:fs} holds for each agent-object pair including agents 1 or 2 (since they have the lowest priority for each object) or objects $a,b$, or $e$ (since all agents have the same priority for those).
For the remaining pairs, we note that agents 3,4, and 5 each receive probability $\frac23$ for objects they prefer to $c$ and $d$ and each of $c$ and $d$ is assigned to agents $\{3, 4, 5\}$ with probability $\frac13$.

Next we argue that $p$ is not ex-post stable. 
Any decomposition of $p$ into deterministic allocations must contain a deterministic allocation $p$ that assigns object $e$ to agent 5.
Neither 1 nor 2 can receive $c$ or $d$ in $p$, since then $(5,c)$ and $(5,d)$ would be blocking pairs, respectively.
Hence, $p$ assigns $c$ to agent 4 and $d$ to agent 3.
On the other hand, $p(i,a) = p(i,b) = 0$ and so $p$ cannot assign $a$ or $b$ to agent 1. 
It follows that 1 remains unmatched in $p$, which violates stability of $p$.
\end{example}

\begin{example}
	For each of our stability notions $S$, $S$-VER results in a different mechanism.

	Consider again \Cref{ex:fractionalnotex-post}.
	One can check that the allocation $p$ is the outcome of VER for fractional stability.
	Hence, not only is fractional stability different from ex-post stability, but it also leads to different instantiation of VER.
	
	VER for claimwise stability in \Cref{ex:fractionalnotex-post} yields the following allocation.
	\begin{center}
					 \begin{blockarray}{ccccccccccc}
					 		&&\matindex{$a$}&\matindex{$b$}& \matindex{$c$}& \matindex{$d$}& \matindex{$e$}\\
					 	    \begin{block}{c(cccccccccc)}
					 			\matindex{$1$}& &$0$&$0$&$\nicefrac{1}{2}$&$\nicefrac{3}{10}$&$\nicefrac{1}{5}$& \\
					 			\matindex{$2$}& &$0$&$0$&$\nicefrac{1}{2}$&$\nicefrac{3}{10}$&$\nicefrac{1}{5}$& \\
								\matindex{$3$}& &$$\nicefrac{1}{3}$$&$$\nicefrac{1}{3}$$&$0$&$\nicefrac2{15}$ &$\nicefrac{1}{5}$&\\
								\matindex{$4$}& &$\nicefrac{1}{3}$&$\nicefrac{1}{3}$& 0 & $\nicefrac2{15}$&$\nicefrac{1}{5}$&\\
								\matindex{$5$}& &$\nicefrac{1}{3}$&$\nicefrac{1}{3}$& 0 &$\nicefrac2{15}$ &$\nicefrac{1}{5}$&\\
					 	    \end{block}
					 	  \end{blockarray}
						  \end{center}
	Thus, VER for fractional stability and claimwise stability are different.
	
		The following example shows that EAS-VER is different from ExpS-VER.
	Let $N=\{1,\dots,8\}$, $O=\{a,\dots,h\}$, and $\capa(i) = 1$ for all $i\in N$. 
	We truncate the preferences to the part that is relevant for computing ExpS-VER.
	All agents have the same priority for objects other than $e$.
	\begin{center}
	\begin{tabular}{lccccc}
	$\succ_1$:&$a$&$e$&$g$&\\
	$\succ_2$:&$a$&$f$&$g$&\\
	$\succ_3$:&$b$&$e$&$g$&\\
	$\succ_4$:&$b$&$f$&$g$&\\
	$\succ_5$:&$c$&$e$&$h$&\\
	$\succ_6$:&$c$&$f$&$h$&\\
	$\succ_7$:&$d$&$e$&$h$&\\
	$\succ_8$:&$d$&$f$&$h$&\\
	\end{tabular}
	\quad\quad
	\begin{tabular}{lccc}
	$\succsim_e$:&  1,3 & 5,7 &2,4,6,8\\
	\end{tabular}
	\end{center}
	ExpS-VER (as well as the probabilistic serial rule) yield the allocation shown below.
	Note that $p$ is not ex-ante stable as agent $1$ wants to get more of $e$ and has higher priority for $e$ than $5$ and $7$.
		\begin{center}$p=$
						 \begin{blockarray}{ccccccccccc}
	&&\matindex{$a$}&\matindex{$b$}& \matindex{$c$}& \matindex{$d$}& \matindex{$e$}& \matindex{$f$}& \matindex{$g$}& \matindex{$h$} \\
						 	    \begin{block}{c(cccccccccc)}
	\matindex{$1$}& &$\nicefrac{1}{2}$&$0$&$0$&$0$&$\nicefrac{1}{4}$&$0$&$\nicefrac{1}{4}$&$0$&\\
	\matindex{$2$}&&$\nicefrac{1}{2}$&$0$&$0$&$0$&$0$&$\nicefrac{1}{4}$&$\nicefrac{1}{4}$&$0$&\\
	\matindex{$3$}& &$0$&$\nicefrac{1}{2}$&$0$&$0$&$\nicefrac{1}{4}$&$0$&$\nicefrac{1}{4}$&$0$&\\
	\matindex{$4$}& &$0$&$\nicefrac{1}{2}$&$0$&$0$&$0$&$\nicefrac{1}{4}$&$\nicefrac{1}{4}$&$0$&\\
	\matindex{$5$}&&$0$&$0$&$\nicefrac{1}{2}$&$0$&$\nicefrac{1}{4}$&$0$&$0$&$\nicefrac{1}{4}$&\\
	\matindex{$6$}&&$0$&$0$&$\nicefrac{1}{2}$&$0$&$0$&$\nicefrac{1}{4}$&$0$&$\nicefrac{1}{4}$&\\
	\matindex{$7$}&&$0$&$0$&$0$&$\nicefrac{1}{2}$&$\nicefrac{1}{4}$&$0$&$0$&$\nicefrac{1}{4}$&\\
	\matindex{$8$}&&$0$&$0$&$0$&$\nicefrac{1}{2}$&$0$&$\nicefrac{1}{4}$&$0$&$\nicefrac{1}{4}$&\\
	\end{block}
						 	  \end{blockarray}
							  \end{center}
	To see that $p$ is ex-post stable, we observe that it can be written as the uniform convex combination of the following four deterministic stable allocations.
	\begin{center}
		\begin{blockarray}{cccccccc}
		1 & 2 & 3 & 4 & 5 & 6 & 7 & 8\\
		\begin{block}{(cccccccc)}
		$a$ & $f$ & $b$ & $g$ & $e$ & $c$ & $h$ & $d$\\
		\end{block}
		\begin{block}{(cccccccc)}
		$a$ & $g$ & $b$ & $f$ & $h$ & $c$ & $e$ & $d$\\
		\end{block}
		\begin{block}{(cccccccc)}
			$g$ & $a$ & $e$ & $b$ & $c$ & $f$ & $d$ & $h$\\
		\end{block}
		\begin{block}{(cccccccc)}
			$e$ & $a$ & $g$ & $b$ & $c$ & $h$ & $d$ & $f$\\
		\end{block}
		\end{blockarray}
	\end{center}		
	Since $p$ is ex-post stable and coincides with the allocation produced by the probabilistic serial rule, it is also the outcome of ExpS-VER.
	Since $p$ is not ex-ante stable, it cannot be the outcome of EAS-VER.
\end{example}

\section{The Vigilant Priority Rule}\label{app:vp}

\begin{algorithm}[tb]
  	\caption{Vigilant Priority}
  	\label{algo:closedverVP}

  	\renewcommand{\algorithmicrequire}{\wordbox[l]{\textbf{Input}:}{\textbf{Output}:}}
  	\renewcommand{\algorithmicensure}{\wordbox[l]{\textbf{Output}:}{\textbf{Output}:}}
  	\hspace*{\algorithmicindent} \textbf{Input:} A preference profile $\apref$, a non-empty and closed set $X$ of allocations, and a permutation $\sigma$ over $N$\\
  	\hspace*{\algorithmicindent} \textbf{Output:} An allocation $p\in X$
  	\algsetup{linenodelimiter=\,}
  	  \begin{algorithmic}[1]
  	\STATE $\pi(i,E)\longleftarrow 0$ for all $i\in N$ and $E\in \mathcal E_i$
	\FOR{$i\in N$ in order of $\sigma$}
	\FOR{$E_i\in\mathcal E_i$ in order of $\pref_i$}
	  	\STATE Compute $\delta^*$ as follows\label{ln:delta_vp}
	  	\begin{align*}
	  	\delta^*&=\max_{\delta,p} \delta  \text{ s.t. }& \\
		p(i,E_i)&\ge  \pi(i,E_i) + \delta& &&\text{(agent $i$ eats $E_i$)}\\
	  	p(j,E)&\ge \pi(j,E) &\forall j\in N\text{ and } E\in \mathcal E_j &&\text{($p$ extends $\pi$)}\\
	  	p&\in X &&&\text{($p$ is feasible)}
	  	\end{align*}
		\vspace{-1em}
  	\STATE $\pi(i,E_i)\longleftarrow \pi(i,E_i)+\delta^*$ 
	\ENDFOR
	\ENDFOR
  	\RETURN $p$ (as computed in Line~\ref{ln:delta_vp})
  	\end{algorithmic}
\end{algorithm}

\end{document}